\documentclass[11pt]{article}
\usepackage{amsmath,amssymb, amsthm}
\usepackage{graphicx}
\usepackage{xfrac}
\usepackage{microtype}
\usepackage{color}
\usepackage{mathptmx}
\usepackage[square,numbers]{natbib}

\usepackage{fullpage}

\usepackage{mathrsfs}  

\newtheorem{theorem}{Theorem}
\newtheorem{lemma}[theorem]{Lemma}
\newtheorem{proposition}[theorem]{Proposition}
\newtheorem{corollary}[theorem]{Corollary}

\newenvironment{remark}[1][Remark]{\begin{trivlist}
\item[\hskip \labelsep {\bfseries #1}]}{\end{trivlist}}

\newenvironment{definition}[1][Definition.]{\begin{trivlist}
\item[\hskip \labelsep {\bfseries #1}]}{\end{trivlist}}

\newenvironment{fact}[1][{\bf Fact}.]{\begin{trivlist}
\item[\hskip \labelsep {\bfseries #1}]}{\end{trivlist}}

\newcommand{\remove}[1]{}

\newcommand{\eqdef}{~\stackrel{\mbox{\tiny \textnormal{def}}}{=}~}
\newcommand{\tup}[1]{\left\langle #1 \right\rangle} 
\newcommand{\pr}[2][{}]{\mathrm{Pr}_{#1}\!\left[ #2 \right]}

\newcommand{\expect}[2][{}]{\mathbb{E}_{#1}\!\left[ #2 \right]}
\newcommand{\var}[2][{}]{\mathrm{Var}_{#1}\!\left[ #2 \right]}
\newcommand{\set}[1]{\left\{#1\right\}}

\newcommand{\ket}[1]{\left| #1 \right\rangle} 
\newcommand{\bra}[1]{\left\langle #1 \right|} 

\newcommand{\ints}{\mathbb{Z}}  
\newcommand{\coms}{\mathbb{C}}  

\newcommand{\aut}{\mathrm{Aut}} 
\newcommand{\tr}{\mathrm{Tr}} 
\newcommand{\Ind}{\mathrm{Ind}} 

\newcommand{\GL}{\mathsf{GL}} 
\newcommand{\PGL}{\mathsf{PGL}} 

\newcommand{\FF}{\mathbb{F}}

\title{The McEliece Cryptosystem\\ Resists Quantum Fourier Sampling Attacks}
\author{Hang Dinh\\ 
Indiana University South Bend\\
\textsf{hdinh@cs.iusb.edu} 
\and
Cristopher Moore\\
University of New Mexico\\
and Santa Fe Institute\\
\textsf{moore@cs.unm.edu} 
\and  
Alexander Russell\\
University of Connecticut\\
\textsf{acr@cse.uconn.edu}
}

\begin{document}

\maketitle              
\thispagestyle{empty}

\begin{abstract}
Quantum computers can break the RSA and El Gamal public-key cryptosystems, since they can factor integers and extract discrete logarithms.  If we believe that quantum computers will someday become a reality, we would like to have \emph{post-quantum} cryptosystems which can be implemented today with classical computers, but which will remain secure even in the presence of quantum attacks.  

In this article we show that the McEliece cryptosystem over \emph{well-permuted, well-scrambled} linear codes resists precisely the attacks to which the RSA and El Gamal cryptosystems are vulnerable---namely, those based on generating and measuring coset states.  This eliminates the approach of strong Fourier sampling on which almost all known exponential speedups by quantum algorithms are based.  Specifically, we show that the natural case of the Hidden Subgroup Problem to which the McEliece cryptosystem reduces cannot be solved by strong Fourier sampling, or by any measurement of a coset state.  We start with recent negative results on quantum algorithms for Graph Isomorphism, which are based on particular subgroups of size two, and extend them to subgroups of arbitrary structure, including the automorphism groups of linear codes.  
This allows us to obtain the first rigorous results on the security of the McEliece cryptosystem in the face of quantum adversaries, strengthening its candidacy for post-quantum cryptography.  
\end{abstract}
\newpage
\pagenumbering{arabic}
\section{Introduction}

Considering that common public-key cryptosystems such as RSA and El Gamal are insecure against quantum attacks, the susceptibility of other well-studied public-key systems to such attacks is naturally of fundamental interest. In this article we present evidence for the strength of the McEliece cryptosystem against quantum attacks, demonstrating that the quantum Fourier sampling attacks that cripple RSA and El Gamal do not apply to the McEliece system coupled with \emph{well-permuted}, \emph{well-scrambled} linear codes. While our results do not rule out other quantum (or classical) attacks, they do demonstrate security against the hidden subgroup methods that have proven so powerful for computational number theory.  Additionally, we partially extend results of~\citet{Ref_Kempe07permutation} concerning the subgroups of $S_n$ reconstructible by quantum Fourier sampling. 

\paragraph{The McEliece cryptosystem.} This public-key cryptosystem was proposed by McEliece in 1978 \citep{Ref_McEliece78public}, and is typically built over Goppa codes. 
There are two basic types of attacks known against the McEliece cryptosystem:  ciphertext only attacks, and attacks on the private key. The former is unlikely to work because it relies on solving the general decoding problem, which is NP-hard. The latter can be successful on certain classes of  linear codes, and is our focus. 
In the McEliece cryptosystem, the private key of a user Alice consists of three matrices:  a $k\times n$ generator matrix $M$ of a hidden $q$-ary $[n,k]$-linear code, an invertible $k\times k$ matrix $A$ over the finite field $\FF_q$, and  an $n\times n$ permutation matrix $P$. Both matrices $A$ and $P$ are selected randomly. Alice's public key includes the matrix $M^*=AMP$, which is a generator matrix of a linear code equivalent to the secret code. An adversary  may attack  the private key by first computing the secret generator matrix $M$, and then computing\footnote{Recovering the secret scrambler and the secret permutation is different from the Code Equivalence problem. The former finds a transformation between two equivalent codes, while the latter decides whether two linear codes are equivalent.} the secret row ``scrambler'' $A$ and the secret permutation $P$. 

There have been some successful attacks on McEliece-type public-key systems.
A notable one is Sidelnokov and Shestakov's attack \cite{Ref_Sidelnikov92insecurity}, which efficiently computes the matrices $A$ and $MP$ from the public matrix $AMP$, in the case that the secret code is a generalized Reed-Solomon (GRS) code. Note that this attack does not reveal the secret permutation.  An attack in which the secret permutation is revealed was proposed by Loidreau and Sendrier \cite{Ref_Loidreau01weak}. However, this attack  only works with a very limited subclass of classical binary Goppa codes, namely those with a binary generator polynomial.

Although the McEliece cryptosystem is efficient and still considered (classically) secure \cite{Ref_Engelbert07summary}, it is rarely used in practice because of the comparatively large public key (see remark 8.33 in~\citep{Ref_Menezes96handbook}). The discovery of successful quantum attacks on RSA and El Gamal, however, have changed the landscape: as suggested by~\citet{Ref_Ryan07excluding} and Bernstein et al. \cite{Ref_Bernstein08attacking}, the McEliece cryptosystem could become a ``post-quantum'' alternative to RSA.

\paragraph{Quantum Fourier sampling.} Quantum Fourier Sampling (QFS) is a key ingredient in most efficient algebraic quantum algorithms, including Shor's algorithms for factorization and discrete logarithm~\citep{Ref_Shor97polynomial} and Simon's algorithm~\citep{Ref_Simon97power}. In particular, Shor's algorithm relies on quantum Fourier sampling over the cyclic group $\ints_N$, while Simon's algorithm uses quantum Fourier sampling over $\ints_2^n$.  In general, these algorithms solve instances of the \emph{Hidden Subgroup Problem} (HSP) over a finite group $G$.  Given a function $f$ on $G$ whose level sets are left cosets of some unknown subgroup $H<G$, i.e., such that $f$ is constant on each left coset of $H$ and distinct on different left cosets, they find a set of generators for the subgroup $H$. 

The standard approach to this problem treats $f$ as a black box and applies $f$ to a uniform superposition over $G$, producing the coset state $\ket{cH}=(\sfrac{1}{\sqrt{|H|}})\sum_{h\in H}\ket{ch}$ for a random $c$.  We then measure $\ket{cH}$ in a Fourier basis $\set{\ket{\rho, i,j}}$ for the space $\coms[G]$, where $\rho$ is an irrep\footnote{Throughout the paper, we write ``irrep'' as short for ``irreducible representation''.} of $G$ and $i,j$ are row and column indices of a matrix $\rho(g)$.  In the \emph{weak} form of Fourier sampling, only the representation name $\rho$ is measured, while in the \emph{strong} form, both the representation name and the matrix indices are measured. This produces probability distributions from which classical information can be extracted to recover the subgroup $H$.  
Moreover, since $\ket{cH}$ is block-diagonal in the Fourier basis, the optimal measurement of the coset state can always be described in terms of strong Fourier sampling.

Understanding the power of Fourier sampling in nonabelian contexts has been an ongoing project, and a sequence of negative results \citep{Ref_Grigni04quantum,Ref_Moore08symmetric,Ref_Hallgren06limitations} have suggested that the approach is inherently limited when the underlying groups are rich enough. In particular, Moore, Russell, and Schulman~\citep{Ref_Moore08symmetric} showed that over the symmetric group, even the strong form of Fourier sampling cannot efficiently distinguish the conjugates of most order-2 subgroups from each other or from the trivial subgroup. That is, 
for any $\sigma \in S_n$ with large support, and most $\pi \in S_n$, 
if $H = \{ 1, \pi^{-1} \sigma \pi \}$ then strong Fourier sampling, and therefore any measurement we can perform on the coset state, 
yields a distribution which is exponentially close to the distribution corresponding to $H=\{1\}$. This result implies that the \textsc{Graph Isomorphism} cannot be solved by the naive reduction to strong Fourier sampling. \citet{Ref_Hallgren06limitations} strengthened these results, demonstrating that even entangled measurements on $o(\log n!)$ coset states result in essentially information-free outcome distributions.  \citet{Ref_Kempe05hidden}  
showed that weak Fourier sampling single coset states in $S_n$ cannot 
distinguish the trivial subgroup from larger subgroups $H$ with polynomial size and non-constant minimal degree.\footnote{The minimal degree of a permutation group $H$ is the minimal number of points moved by a non-identity element of $H$.} 
They conjectured, conversely, that if a subgroup $H<S_n$ can be distinguished from the trivial subgroup by weak Fourier sampling, then the minimal degree of $H$ must be constant.  Their conjecture was later proved by Kempe, Pyber, and Shalev \citep{Ref_Kempe07permutation}.

\subsection{Our contributions} To state our results, we say that a subgroup $H<G$ is \emph{indistinguishable by strong Fourier sampling} 
if the conjugate subgroups $g^{-1}Hg$ cannot be distinguished from each other or from the trivial subgroup by measuring the coset state in an arbitrary Fourier basis.  A precise definition is presented in Section~\ref{Sec:General}.  
Since the optimal measurement of a coset state can always be expressed as an instance of strong Fourier sampling, these results imply that no measurement of a single coset state yields any useful information about $H$.
Based on the strategy of Moore, Russell, and Schulman~\citep{Ref_Moore08symmetric}, we first develop a general framework, formalized in Theorem~\ref{Thm:general}, to determine indistinguishability of a subgroup by strong Fourier sampling.  We emphasize that their results cover the case where the subgroup has order two.  Our principal contribution is to show how to extend their methods to more general subgroups.  

We then apply this general framework to a class of semi-direct products $(\GL_k(\FF_q)\times S_n)\wr \ints_2$, bounding  the distinguishability \remove{by strong Fourier sampling} for the HSP corresponding to the private-key attack on the McEliece cryptosystem, i.e., the problem of determining $A$ and $P$ from $M^*$ and $M$. Our bound, given in Corrolary \ref{Cor:McEliece} of Theorem \ref{Thm:McEliece}, depends on  the minimal degree and the size of the automorphism group of the secret code,  
as well as on  the column rank of the secret generator matrix. In particular, the rational Goppa codes have good values for these quantities, i.e., they have small automorphism groups with large minimal degree, and have generator matrices of full rank.
In general, our result indicates that the McEliece cryptosystem resists all known attacks based on strong Fourier sampling if its secret $q$-ary $[n,k]$-code \textit{(i)} is \emph{well-permuted}, i.e., its automorphism group has minimal degree $\Omega(n)$ and size $e^{o(n)}$, and \textit{(ii)} is \emph{well-scrambled}, i.e., it has a generator matrix of rank at least $k-o(\sqrt{n})$. Here, we assume $q^{k^2}\leq n^{0.2 n}$,  
which implies $\log |\GL_k(\FF_q)|=O(n\log n)$, so that Alice only needs to flip $O(n\log n )$ bits to pick a random matrix $A$ from $\GL_k(\FF_q)$. Thus she needs only $O(n\log n)$ coin flips overall to generate her private key. 

While our main application is the security of the McEliece cryptosystem, we show in addition that our general framework is applicable to other classes of groups with simpler structure, including the symmetric group and the finite general linear group $\GL_2(\mathbb{F}_q)$.  
For the symmetric group, we extend the results of \citep{Ref_Moore08symmetric} to larger subgroups of $S_n$. Specifically, we show that any subgroup $H<S_n$ with minimal degree  $m\geq \Theta(\log |H|) +\omega(\log n)$ is indistinguishable by strong Fourier sampling over $S_n$. 
In other words, if the conjugates of $H$ can be distinguished from each other---or from the trivial subgroup---by strong Fourier sampling, then the minimum degree of $H$ must be $O(\log|H|) + O(\log n)$.  This partially extends the results of \citet{Ref_Kempe07permutation}, which apply only to weak Fourier sampling.

We go on to demonstrate another  application of our general framework for the general linear group $\GL_2(\mathbb{F}_q)$, giving the first negative result regarding the power of strong Fourier sampling  
over $\GL_2(\mathbb{F}_q)$. We show that any subgroup $H< \GL_2(\mathbb{F}_q)$ that does not contain non-identity scalar matrices and has order $|H| \leq q^{\delta}$ for some  $\delta<1/2$ is indistinguishable by strong Fourier sampling. 
Examples of such subgroups are those generated by a constant number of triangular unipotent matrices.

\begin{remark}
Our results show that the natural reduction of McEliece to a hidden subgroup problem yields negligible information about the secret key.  Thus they rule out the direct analogue of the quantum attack that breaks, for example, RSA.  Our results are quite flexible in this hidden-subgroup context: they apply equally well to any HSP reduction resulting in a rich subgroup of $\GL_2(\mathbb{F}_q)$, which seems to be the natural arena for the McEliece system.  

Of course, our results do not rule out other quantum (or classical) attacks. Neither do they establish that a quantum algorithm for the McEliece cryptosystem would violate a natural hardness assumption, as do recent lattice cryptosystem constructions whose hardness is based on the Learning With Errors problem (e.g. \citet{regev-jacm}).  Nevertheless, they indicate that any such algorithm would have to use rather different ideas than those that have been proposed so far.
%
\end{remark}

\subsection{Summary of technical ideas} 
Let $G$ be a finite group.  We wish to establish general criteria for indistinguishability of subgroups $H < G$ by strong Fourier sampling. 
We begin with the general strategy, developed in~\cite{Ref_Moore08symmetric}, that controls the resulting probability distributions in terms of the representation-theoretic properties of $G$. In order to handle richer subgroups, however, we have to overcome some technical difficulties.  Our principal contribution here is a ``decoupling'' lemma that allows us to handle the cross terms arising from pairs of nontrivial group elements.

Roughly, the approach (presented in Section \ref{Sec:General})  identifies two disjoint subsets, \textsc{Small} and \textsc{Large}, of irreps of $G$. The set \textsc{Large} consists of all irreps whose dimensions are no smaller than a certain threshold $D$. While $D$ should be as large as possible, we also need to choose $D$ small enough so that the set \textsc{Large} is large. In contrast, the representations in \textsc{Small} must have small dimension (much smaller than $\sqrt{D}$), and the set \textsc{Small} should be small or contain few irreps that appear in the decomposition of the tensor product representation $\rho\otimes\rho^*$ for any $\rho\in\textsc{Large}$. In addition, any irrep $\rho$ outside $\textsc{Small}$ must have small normalized character $|\chi_{\rho}(h)|/d_{\rho}$ for any nontrivial element $h\in H$.
If there are such two sets \textsc{Small} and \textsc{Large}, and if the order of $H$ is sufficiently small, then $H$ is indistinguishable by strong Fourier sampling over $G$.

In the case $G=\GL_2(\mathbb{F}_q)$, for instance, we choose \textsc{Small} as the set of all linear representations and set the threshold $D=q-1$. The key lemma we need to prove is then that for any nonlinear irrep $\rho$ of $\GL_2(\mathbb{F}_q)$, the decomposition of $\rho\otimes\rho^*$ contains at most two inequivalent linear representations. (Lemma \ref{Lemma:LinearRep_0}).  
In the case $G=S_n$, we choose \textsc{Small} as the set $\Lambda_c$ of all Young diagrams with at least $(1-c)n$ rows or at least $(1-c)n$ columns, and set  $D=n^{dn}$, for reasonable constants $0<c,d<1$. For this case, we use the same techniques as in~\citep{Ref_Moore08symmetric}.

For the case $G=(\GL_k(\FF_q)\times S_n)\wr \ints_2$ corresponding to the McEliece cryptosystem,  the normalized characters on the hidden subgroup $K$ depend on the minimal degree of  the automorphism group $\aut(C)$, where $C$ is the secret code. Moreover,  $|K|$ depends on  $|\aut(C)|$ and the column rank of the secret generator matrix. Now we can choose  \textsc{Small} as the set of all irreps constructed from tensor product representations $\tau\times \lambda$ of  $\GL_k(\FF_q)\times S_n$ with $\lambda\in\Lambda_c$. Then the ``small'' features of $\Lambda_c$ will induce the ``small'' features of this set \textsc{Small}.    To show that any irrep outside \textsc{Small} has small normalized characters on $K$, we show that any Young diagram $\lambda$ outside $\Lambda_c$ has large dimension (Lemma \ref{Lemma:LargeDimS_n}).

\section{Hidden Subgroup Attack Against McEliece Cryptosystems}\label{ApxSec:AttackingMcEliece}
\remove{
The McEliece public-key encryption scheme works by first selecting a particular linear code for which an efficient decoding algorithm is known, and then camouflaging the code as a general linear code. The idea is based on the \textsf{NP}-hardness of decoding an arbitrary linear code. Therefore, a description of the original code can serve as the private key, while a description of the transformed code serves as the public key \citep{Ref_Menezes96handbook}. 

Specifically, in a McEliece cryptosystem with common system parameters given by fixed positive integers $n, k, t$, the private key of an entity Alice is a triple $(A_0,M,P_0)$, where 
\begin{itemize}
\item $M$ is a $k\times n$ generator matrix for a  $q$-ary $[n,k]$-linear code $C$ which can correct $t$ errors, and for which an efficient decoding algorithm is known, 
\item $A_0$ is a $k\times k$ invertible matrix chosen randomly from $\GL_k(\mathbb{F}_q)$, and 
\item $P_0$ is an $n\times n$ permutation matrix corresponding to a permutation chosen randomly from $S_n$. 
\end{itemize}
The public key of Alice then consists of the parameter $t$ and the $k\times n$ matrix ${M}^*=A_0MP_0$, which is a generator matrix for a (general) linear code equivalent to $C$. 
}
\subsection{An attack via the hidden shift problem}


As mentioned in the Introduction, we consider the attack that involves finding the secret scrambler and permutation in a McEliece private key. 
\begin{definition}[Scrambler-Permutation Problem]
Given two $k\times n$  generator matrices $M$ and $M^*$ of two equivalent linear codes over $\FF_q$, the task is to find   a matrix $A\in \GL_k(\FF_q)$ and an $n\times n$ permutation $P$ matrix such that $M^*= AMP$.
\end{definition}

The decision version of this problem, known as \textsc{Code Equivalence} problem, is not easier than  \textsc{Graph Isomorphism}, although it is unlikely to be NP-complete \cite{Ref_Petrank97code}.
The only known way to solve the Scrambler-Permutation problem using quantum Fourier sampling is to reduce it to a Hidden Shift Problem, which in turn can be reduced to a Hidden Subgroup Problem over a wreath product.

\begin{definition}[Hidden Shift Problem]  Let $G$ be a finite group and $\Sigma$ be some finite set. Given two functions $f_0: G\to \Sigma$ and $f_1: G\to \Sigma$ on $G$, we call  an element $s\in G$ a \emph{left shift} from $f_0$ to $f_1$ (or simply, a \emph{shift}) if $f_0(  sx)=f_1(x)$ for all $x\in G$. We are promised that there is such a shift. Find a shift.
\end{definition}

The Scrambler-Permutation Problem is reduced to the Hidden Shift Problem over group $G=\GL_k(\FF_q)\times S_n$ by defining functions $f_0$ and $f_1$ on $\GL_k(\mathbb{F}_q)\times S_n$ as follows: for all $(A,P)\in \GL_k(\FF_q)\times S_n$,
\begin{equation}\label{Eq:McEliece_f_0}
f_0(A,P) = A^{-1}MP\,,
\quad\quad
f_1(A,P) = A^{-1}{M}^*P\,.
\end{equation}
Here and from now on, we identify each $n\times n$ permutation matrix as its corresponding permutation in $S_n$.
Apparently, $AMP=M^*$ if and only if $(A^{-1},P)$ is a shift from $f_0$ to $f_1$. 

\subsection{Reduction from the hidden shift problem to the hidden subgroup problem}

We present how to reduce the Hidden Shift Problem over group $G$ to the HSP on the wreath product $G\wr \ints_2$, which can also be written as a semi-direct product $G^2\rtimes \ints_2$ associated with the action of $\ints_2$ on $G^2$ in which the non-identity element of $\ints_2$ acts on $G^2$ by swapping, i.e.,  $1\cdot (x,y)=(y,x)$. 
\remove{
Hence, the multiplication on $G^2\rtimes \ints_2$ is defined as follows:
\[
((x_1,y_1), 0)((x_2, y_2), b) = ((x_1x_2,y_1y_2), b)
\]
\[
((x_1,y_1), 1)((x_2, y_2), b) = ((x_1y_2,y_1x_2), b+1)\,.
\]
The inversion is given by $((x,y),b)^{-1} = (b\cdot (x^{-1}, y^{-1}),b)$. Hence,
\[
((x,y),0)^{-1} = ((x^{-1}, y^{-1}),0)
\]
\[
((x,y),1)^{-1} = ((y^{-1}, x^{-1}),1)
\]
}

Given two input functions $f_0$ and $f_1$ for a Hidden Shift Problem on $G$, we define the function $f:G^2\rtimes \ints_2\to \Sigma\times \Sigma$ as follows: for $(x,y)\in G^2, b\in\ints_2$,

\begin{equation}\label{Eq:function_f}
f((x,y),b) \eqdef 
\begin{cases}
 (f_0(x),f_1(y)) & \text{ if } b=0\\
 (f_1(y),f_0(x)) & \text{ if } b=1\\
\end{cases}
\end{equation}

\remove{
where $\alpha$ is any hidden shift. Another option for function $f$ is
\begin{equation}\label{Eq:function_f_new}
f((x,y),b) \eqdef 
\begin{cases}
 (f_0(x),f_1(y)) & \text{ if } b=0\\
 (f_0(y),f_1(x)) & \text{ if } b=1\\
\end{cases}
\end{equation}
}

We want to determine the subgroup whose cosets are distinguished by $f$. Recall that a function $f$ on a group $G$ \emph{distinguishes the right cosets} of a subgroup $H<G$ if for all $x,y\in G$,
\[
f(x)=f(y)\iff yx^{-1}\in H \,.
\]

\begin{definition}
Let $f$ be a function on a group $G$. We say that the function $f$ is \emph{injective under right multiplication} if for all $x,y\in G$,
\[
f(x)=f(y)\iff f(yx^{-1})=f(1) \,.
\]
Define the subset $G|_f$ of the group $G$: $$G|_f\eqdef  \set{g\in G\mid f(g)=f(1)}\,.$$
\end{definition}
\remove{
\begin{fact}
Let $f$ be a function on a group $G$. If $f$ is injective under right multiplication then 
\begin{enumerate}
\item The subset $G|_f\eqdef  \set{g\in G\mid f(g)=f(1)}$ is a subgroup of $G$. 
\item $f(x)=f(y) \Rightarrow f(xg)=f(yg) ~\forall g\in G$.
\item $f(x)=f(y) \Leftarrow f(xg)=f(yg)$ for some  $g\in G$.
\end{enumerate}
\end{fact}
}
\begin{proposition}
Let $f$ be a function on a group $G$. If $f$ distinguishes the right cosets of a subgroup $H<G$, then $f$ must be injective under right multiplication and $G|_f=H$.
Conversely, if $f$ is injective under right multiplication, then $G|_f$ is a subgroup and $f$ distinguishes the right cosets of the subgroup $G|_f$.
\end{proposition}

Hence, the function $f$ defined in \eqref{Eq:function_f} can distinguish the right cosets of some subgroup if and only if it is injective under right multiplication. 

\begin{lemma}
The function $f$ defined in \eqref{Eq:function_f} is injective under right multiplication if and only if $f_0$ is injective under right multiplication.
\end{lemma}
The proof for this lemma is straightforward on the case by case basis, so we omit it here.

\begin{proposition}
Assume $f_0$ is injective under right multiplication. Let $H_0=G|_{f_0}$ and $s$ be a shift. Then the function $f$ defined in \eqref{Eq:function_f} distinguishes right cosets of the following subgroup of $G^2\rtimes \ints_2$:
\[
G^2\rtimes \ints_2|_f = ((H_0, s^{-1}H_0 s ),0) \cup ((H_0 s, s^{-1}H_0  ),1)
\]
which has size $2 |H_0|^2$. Recall that the set of shifts is $H_0s$.
\end{proposition}

To find a hidden shift from the hidden subgroup $K=G^2\rtimes \ints_2|_f$, just select an element of the form $((g_1,g_2),1)$ from $K$, then $g_1$ must belong to $H_0s$, which is the set of all shifts.

\paragraph{In the case of Scrambler-Permutation problem.} Back to the Hidden Shift Problem over $G=\GL_k(\FF_q)\times S_n$ reduced from the Scrambler-Permutation problem, it is clear that the input function $f_0$  defined in \eqref{Eq:McEliece_f_0} is injective under right multiplication and that
\[
H_0 = \GL_k(\mathbb{F}_2)\times S_n|_{f_0} = \set{(A,P)\in \GL_k(\mathbb{F}_2)\times S_n: A^{-1}MP=M}\,.
\]


\section{Quantum Fourier sampling (QFS)}\label{SubSec:QFS}
\subsection{Preliminaries and Notation}

Fix a finite group G, abelian or non-abelian, and let $\widehat{G}$ denote the set of (complex) irreducible representations, or ``irreps'' for short, of $G$.  For each irrep $\rho\in\widehat{G}$, let $V_{\rho}$ denote a vector space over $\coms$ on which $\rho$ acts so that $\rho$ is a group homomorphism from $G$ to the general linear group over $V_{\rho}$, and let $d_{\rho}$ denote the dimension of $V_{\rho}$.  
For each $\rho$, we fix an orthonormal basis $B_{\rho}=\set{\mathbf{b}_1,\ldots,\mathbf{b}_{d_{\rho}}}$ for $V_{\rho}$, in which we can represent each $\rho(g)$ as a $d_{\rho}\times d_{\rho}$ unitary matrix whose $j^{\rm th}$ column is the vector $\rho(g)\mathbf{b}_j$.

Viewing the vector space $\coms[G]$ as the regular representation of $G$, we can decompose $\coms[G]$ into irreps as the direct sum $\bigoplus_{\rho\in\widehat{G}}V_{\rho}^{\oplus d_{\rho}}$.  This has a basis
\(
\{\ket{\rho,i,j}: \rho\in\widehat{G}, 1\leq i,j\leq d_{\rho}\}
\),
where $\{\ket{\rho,i,j}\mid 1\leq i\leq d_{\rho}\}$ is a  basis for the $j^{\rm th}$ copy of $V_{\rho}$ in the decomposition of $\coms[G]$.

\begin{definition}The \emph{Quantum Fourier transform} over $G$ is the unitary operator, denoted $F_G$, that transforms a vector in $\coms[G]$ from the point-mass basis  $\set{\ket{g}\mid g\in G}$ into the basis given by the decomposition of $\coms[G]$.  For all $g\in G$,
\[
F_G\ket{g} = \sum_{\rho, i,j}\sqrt{\frac{d_{\rho}}{|G|}}\rho(g)_{i,j} \,\ket{\rho,i,j}\,,
\]
where $\rho(g)_{ij} $ is the $(i,j)$-entry of the matrix $\rho(g)$. Alternatively, we can view $F_G\ket{g}$ as a block diagonal matrix consisting of the block $\sqrt{{d_{\rho}}/{|G|}} \,\rho(g)$ for each $\rho\in\widehat{G}$.
\end{definition}

\paragraph{Notations.} For each subset $X\subset G$, define $\ket{X}=(\sfrac{1}{\sqrt{|X|}})\sum_{x\in X}\ket{x}$, which is the state of uniformly random element of $X$ in the point-mass basis. For each $X\subset G$ and $\rho\in\widehat{G}$, define the operator
\[
\Pi^{\rho}_X \eqdef \frac{1}{|X|}\sum_{x\in X}\rho(x)\,,
\]
and let $\widehat{X}(\rho)$ denote the $d_{\rho}\times d_{\rho}$ matrix block at $\rho$ in the quantum Fourier transform of $\ket{X}$, i.e.,
\[
\widehat{X}(\rho) \eqdef \sqrt{\frac{d_{\rho}}{|G||X|}}\sum_{x\in X}\rho(x)=\sqrt{\frac{d_{\rho}|X|}{|G|}}\Pi^{\rho}_X\,.
\]
\begin{fact}
If $X$ is a subgroup of $G$, then $\Pi^{\rho}_X$ is 
a projection operator.  That is, $(\Pi^{\rho}_X)^{\dagger}=\Pi^{\rho}_X$ and $(\Pi^{\rho}_X)^2=\Pi^{\rho}_X$.
\end{fact}

Quantum Fourier Sampling (QFS) is a standard procedure based on the Quantum Fourier Transform to solve the Hidden Subgroup Problem (HSP) (see \citep{Ref_Lomont04hidden} for a survey).  An instance of the HSP over $G$ consists of a black-box function $f:G\to\set{0,1}^*$ such that $f(x)=f(y)$ if and only if $x$ and $y$ belong to the same left coset of $H$ in $G$, for some subgroup $H\leq G$.  The problem is to recover $H$ using the oracle $O_f: \ket{x,y}\mapsto \ket{x,y\oplus f(x)}$. The general QFS procedure for this is the following:
  
\begin{enumerate}
\item Prepare a 2-register quantum state, the first in a uniform superposition of the group elements and the second with the value zero:
\(\ket{\psi_1} = (\sfrac{1}{\sqrt{|G|}})\sum_{g\in G}\ket{g}\ket{0}\,.
\)
\item Query $f$, i.e., apply the oracle $O_f$, resulting in the state
\[
\ket{\psi_2} = O_f\ket{\psi_1}=\frac{1}{\sqrt{|G|}}\sum_{g\in G}\ket{g}\ket{f(g)}
=\frac{1}{\sqrt{|T|}}\sum_{\alpha\in T}\ket{\alpha H}\ket{f(\alpha)}
\] 
where $T$ is a transversal of $H$ in $G$. 

\item\label{QFS_Step_Coset} Measure the second register of  $\ket{\psi_2}$, resulting in the state $\ket{\alpha H}\ket{f(\alpha)}$ with probability $1/|T|$ for each $\alpha\in T$.  The first register of the resulting state is then $\ket{\alpha H}$ for some uniformly random $\alpha\in G$.

\remove{
Alternatively, we can view the first register of the resulting state as being in a mixed state with density matrix:
\[
\mathcal{M}_H = \frac{1}{|T|}\sum_{\alpha\in T} \ket{\alpha H}\bra{\alpha H}=\frac{1}{|G|}\sum_{g\in G}\ket{g H}\bra{g H}\,.
\]
}
\item\label{QFS_Step_Fourier} Apply the quantum Fourier transform over $G$ to the coset state $\ket{\alpha H}$ observed at step \ref{QFS_Step_Coset}:
\[
F_G\ket{{\alpha H}}  =
\sum_{\rho\in\widehat{G}, 1\leq i,j\leq d_{\rho}}\widehat{\alpha H}(\rho)_{i,j}\ket{\rho, i,j}\,.
\]
\item\label{QFS_Step_Measure} (Weak) Observe the representation name $\rho$.  (Strong) Observe $\rho$ and matrix indices $i,j$.
\item Classically process the information observed from the previous step to determine the subgroup $H$.
\end{enumerate}
\remove{
As discussed in \citep{Ref_Lomont04hidden}, there are several potential difficulties in making this algorithm efficient,  including how to efficiently compute the quantum Fourier transform on $G$, how to choose a basis for each irrep $\rho\in\widehat{G}$, and 
how to efficiently reconstruct the subgroup generators for $H$ from the irreps returned.
Here, we shall focus on the probability distributions resulted by quantum Fourier sampling in both weak and strong forms.
}
\paragraph{Probability distributions produced by QFS.}
For a particular coset $\alpha H$, the probability of measuring the representation $\rho$ in the state $F_G\ket{{\alpha H}}$ is
\[
P_{\alpha H}(\rho) = \|\widehat{\alpha H}(\rho)\|_F^2 
=\frac{d_{\rho}|H|}{|G|} {\tr\left((\Pi^{\rho}_{\alpha H})^{\dagger}\Pi^{\rho}_{\alpha H}\right)}
=\frac{d_{\rho}|H|}{|G|} {\tr\left(\Pi^{\rho}_{H}\right)}
\]
where $\tr(A)$ denotes the trace of a matrix $A$, and $\|A\|_F:=\sqrt{\tr(A^{\dagger}A)}$ is the Frobenius norm of $A$. The last equality is due to the fact that $\Pi^{\rho}_{\alpha H}=\rho(\alpha)\Pi^{\rho}_{H}$ and that $\Pi^{\rho}_{H}$ is an orthogonal projector.

Since there is no point in measuring the rows \citep{Ref_Grigni04quantum}, we are only concerned with measuring the columns. As pointed out in \citep{Ref_Moore08symmetric}, 
the optimal von Neumann measurement on a coset state can always be expressed in this form for some basis $B_{\rho}$.  
Conditioned on observing $\rho$ in the state $F_G\ket{\alpha H}$, the probability of measuring a given $\mathbf{b}\in B_{\rho}$ is $\|\widehat{\alpha H}(\rho)\mathbf{b}\|^2$. Hence the conditional probability that we observe the vector $\mathbf{b}$, given that we observe the representation $\rho$, is then
\[
P_{\alpha H}(\mathbf{b}\mid \rho) = \frac{\|\widehat{\alpha H}(\rho)\mathbf{b}\|^2}{P_{\alpha H}(\rho)}
=\frac{\|\Pi^{\rho}_{\alpha H}\mathbf{b}\|^2}{\tr\left(\Pi^{\rho}_{H}\right)}
=\frac{\|\Pi^{\rho}_{H}\mathbf{b}\|^2}{\tr\left(\Pi^{\rho}_{H}\right)}
\]
where in the last equality, we use the fact that as $\rho(\alpha)$ is unitary, it preserves the norm of the vector $\Pi^{\rho}_{H}\mathbf{b}$. 

The coset representative $\alpha$ is unknown and is uniformly distributed in $T$.  However, both distributions $P_{\alpha H}(\rho)$ and $P_{\alpha H}(\mathbf{b}\mid\rho)$ are independent of $\alpha$ and are the same as those for the state $F_G\ket{H}$. Thus, in Step \ref{QFS_Step_Measure} of the QFS procedure above, we observe $\rho\in\widehat{G}$ with probability
\(
P_H(\rho) 
\), and conditioned on this event, we observe $\mathbf{b}\in B_{\rho}$ with probability $P_{H}(\mathbf{b}\mid\rho)$.


If the hidden subgroup is trivial, $H=\{1\}$, the conditional probability distribution on $B_{\rho}$ is uniform, 
\[
P_{\set{1}}(\mathbf{b}\mid\rho) = \frac{\|\Pi_{\set{1}}^{\rho}\mathbf{b}\|^2}{\tr\left(\Pi_{\set{1}}^{\rho}\right)} = \frac{\|\mathbf{b}\|^2}{d_{\rho}} = \frac{1}{d_{\rho}} \, .
\]



\subsection{Distinguishability by QFS}\label{Sec:General}

We fix a finite group $G$ and consider quantum Fourier sampling over $G$ in the basis given by $\{ B_{\rho} \}$. For a subgroup $H<G$ and for $g\in G$, let $H^g$ denote the conjugate subgroup $g^{-1}Hg$.  
Since $\tr\left(\Pi^{\rho}_{H}\right)=\tr\left(\Pi^{\rho}_{H^g}\right)$, the probability distributions obtained by QFS for recovering the hidden subgroup $H^g$ are  
\[
P_{H^g}(\rho) = \frac{d_{\rho}|H|}{|G|} {\tr\left(\Pi^{\rho}_{H}\right)} = P_H(\rho)
\quad\text{and}\quad
P_{H^g}(\mathbf{b}\mid\rho) = \frac{\|\Pi^{\rho}_{H^g}\mathbf{b}\|^2}{\tr\left(\Pi^{\rho}_{H}\right)}\,.
\]

As $P_{H^g}(\rho)$ does not depend on $g$, weak Fourier sampling can not distinguish conjugate subgroups. Our goal is to point out that for certain nontrivial subgroup $H< G$, strong Fourier sampling  can not efficiently distinguish the conjugates of $H$ from each other or from the trivial one. 
Recall that the distribution $P_{\set{1}}(\cdot\mid\rho)$ obtained by performing strong Fourier sampling  on the trivial hidden subgroup is the same as the uniform distribution $U_{B_{\rho}}$ on the basis $B_{\rho}$. Thus, our goal can be boiled down to showing that the probability distribution $P_{H^g}(\cdot\mid\rho)$ is likely to be  close to the uniform distribution $U_{B_{\rho}}$ in total variation, for a random $g\in G$ and an irrep $\rho\in\widehat{G}$ obtained by weak Fourier sampling. 

\begin{definition}
We define the \emph{distinguishability} of a subgroup $H$ (using strong Fourier sampling over $G$), denoted $\mathcal{D}_H$, to be the expectation of the squared $L_1$-distance between $P_{H^g}(\cdot\mid\rho)$ and $U_{B_{\rho}}$:
\[
\mathcal{D}_H \eqdef \expect[\rho,g]{\|P_{H^g}(\cdot\mid\rho)-U_{B_{\rho}}\|_{1}^2} \,,
\]
where $\rho$ is drawn from $\widehat{G}$ according to the distribution $P_H(\rho)$, and $g$ is chosen from  $G$ uniformly at random. We say that the subgroup $H$ is \emph{indistinguishable} if $\mathcal{D}_H\leq \log^{-\omega(1)} |G|$.  
\end{definition}

Note that if $\mathcal{D}_H$ is small, then the total variation distance between $P_{H^g}(\cdot\mid\rho)$ and $U_{B_{\rho}}$ is small with high probability due to Markov's inequality: for all $\epsilon>0$,
\[
\pr[g]{\|P_{H^g}(\cdot\mid\rho)-U_{B_{\rho}}\|_{t.v.}\geq \epsilon/2} 
=\pr[g]{\|P_{H^g}(\cdot\mid\rho)-U_{B_{\rho}}\|_{1}^2\geq \epsilon^2} 
\leq \mathcal{D}_H/\epsilon^2\,. 
\]
In particular, if the subgroup $H$ is indistinguishable by strong Fourier sampling, then for all constant $c>0$,
\[
\|P_{H^g}(\cdot\mid\rho)-U_{B_{\rho}}\|_{t.v.}< \log^{-c} |G| 
\]
with probability at least $1-\log^{-c} |G|$ in both $g$ and $\rho$. Indeed, our notion of indistinguishability is inspired by that of Kempe and Shalev~\citep{Ref_Kempe05hidden}.  Focusing on weak Fourier sampling, they say that $H$ is indistinguishable if $\|P_H(\cdot)-P_{\set{1}}(\cdot)\|_{t.v.}<\log^{-\omega(1)} |G|$.

\remove{
\begin{remark}
For any distributions $P$ and $Q$, the total variance $\|P-Q\|_{t.v.}$ equals the maximal difference $|P(A)-Q(A)|$ for all event $A$, which implies $\|P-Q\|_{1}=2\|P-Q\|_{t.v.}\leq 2$. Therefore, $\mathcal{D}_H\leq 4$ for every subgroup  $H$. 
\end{remark}
}

Our main theorem below will serve as a general guideline for bounding the distinguishability of $H$.  For this bound, we define, for each $\sigma\in\widehat{G}$, the \emph{maximal normalized character of $\sigma$ on $H$} as 
\[
\overline{\chi}_{\sigma}(H) \eqdef \max_{h\in H\setminus\set{1}} \frac{|\chi_{\sigma}(h)|}{d_{\sigma}}\,.
\] 
For each subset $S\subset \widehat{G}$, let
\[
\overline{\chi}_{\overline{S}}(H) = \max_{\sigma \in \widehat{G}\setminus S}\overline{\chi}_{\sigma}(H) 
\quad\text{and}\quad d_{S}=\max_{\sigma\in S} d_{\sigma}\,.
\]
In addition, for each reducible representation $\rho$ of $G$, we let $I(\rho)$ denote the set of irreps of $G$ that appear in the decomposition of $\rho$ into irreps. 

\begin{theorem}\label{Thm:general}(\textsc{Main Theorem})
Suppose $S$ is a subset of $\widehat{G}$. 
Let $D>d_S^2$ and  $L=L_D\subset\widehat{G}$ be the set of all irreps of  dimension at least $D$. Let
\begin{equation}\label{Eq:general_delta}
\Delta = \Delta_{S,L}=\max_{\rho\in L}{\bigl|S\cap I(\rho\otimes \rho^*)\bigr|}\,.
\end{equation} 
Then the distinguishability of $H$ is bounded by
\[
\mathcal{D}_H\leq  4|H|^2\left(\overline{\chi}_{\overline{S}}(H)+\Delta\frac{d_S^2}{D}+\frac{|\overline{L}|D^2}{|G|}\right)\,.
\]

\end{theorem}
Intuitively, the set $S$ consists of irreps of small dimension, and $L$ consists of irreps of large dimension. Moreover, we wish to have that the size of $S$ is small while the size of $L$ is large, so that most irreps are likely in $L$. In the cases where there are relatively few irreps, i.e. $|S|\ll D$ and $|\widehat{G}|\ll|G|$, we can simply upper bound $\Delta$ by $|S|$ and upper bound $|\overline{L}|$ by $|\widehat{G}|$.

We discuss the proof of this theorem in Section \ref{Sec:MainTheoremProof}.  
\section{Applications}
In this section, we point out some applications of Theorem \ref{Thm:general} to analyze strong Fourier sampling over certain non-abelian groups.

\subsection{Strong Fourier sampling over $S_n$}\label{SubSec:Symmetric}

In this part, we consider the case where $G$ is the symmetric group $S_n$. 
Recall that each irrep of $S_n$ is one-to-one corresponding to an integer partition $\lambda=(\lambda_1,\lambda_2,\ldots,\lambda_{t})$ of $n$, which is associated with the Young diagram of $t$ rows in which the $i^{\rm th}$ row contains $\lambda_i$ columns. The conjugate representation of $\lambda$ is the irrep corresponding to the partition $\lambda'=(\lambda'_1,\lambda'_2,\ldots,\lambda'_{t'})$, which is obtained by flipping the Young diagram $\lambda$ about the diagonal. In particular, $\lambda'_1=t$ and $t'=\lambda_1$.

As in \citep{Ref_Moore08symmetric}, we use Roichman's upper bound \citep{Ref_Roichman96upper} on normalized characters.

\begin{theorem}[Roichman's Theorem \citep{Ref_Roichman96upper}] There exist constant $b>0$ and $0< q <1$ so that for $n>4$, for every $\pi\in S_n$, and for every irrep $\lambda$ of $S_n$,
\[
\left|\frac{\chi_{\lambda}(\pi)}{d_{\lambda}}\right| \leq \left(\max\left(q, \frac{\lambda_1}{n},\frac{\lambda'_1}{n}\right)\right)^{b\cdot \mathrm{supp}(\pi)}
\]
where $\mathrm{supp}(\pi)=\#\set{k\in[n]\mid \pi(k)\neq k}$ is the support of $\pi$.
\end{theorem}

This bound works well for unbalanced Young diagrams. In particular, 
for a constant $0<c<1/4$, let $\Lambda_c$ denote the collection of partitions $\lambda$ of $n$ with the property that either $\frac{\lambda_1}{n}\geq 1-c$ or $\frac{\lambda'_1}{n}\geq 1-c$, i.e., the Young diagram $\lambda$ contains at least $(1-c)n$ rows or contains at least $(1-c)n$ columns. Then, Roichman's upper bound implies that for every $\pi\in S_n$ and $\lambda\not\in \Lambda_c$, and a universal constant $\alpha > 0$, 
\begin{equation}\label{Eq:RoichmanBound_1}
\left|\frac{\chi_{\lambda}(\pi)}{d_{\lambda}}\right| 
\leq e^{-\alpha \cdot\mathrm{supp}(\pi)}\,.
\end{equation}
On the other hand, both $|\Lambda_c|$ and the maximal dimension of representations in $\Lambda_c$ are small, as shown in the following Lemma of \citep{Ref_Moore08symmetric}.

\begin{lemma}[Lemma 6.2 in \citep{Ref_Moore08symmetric}]\label{Lemma:SmallDimS_n} Let $p(n)$ denote the number of integer partitions of $n$. Then $|\Lambda_c|\leq 2cn\cdot p(cn)$, and $d_{\mu}<n^{cn}$ for any $\mu\in \Lambda_c$.
\end{lemma}

To give a more concrete bound for the size of $\Lambda_c$, we record the asymptotic formula for the partition function $p(n)$ \citep[pg. 45]{Ref_Fulton91representation}: 
\(
p(n)\approx {e^{\pi\sqrt{2n/3}}}/{(4\sqrt{3}n)} = e^{O(\sqrt{n})}n^{-1} \text{ as $n\to\infty$}\,.
\)

\remove{
It is natural to ask if there exist irreps of $S_n$ with large dimensions so that the above lemma is applicable. Fortunately, it can be shown that most irreps of $S_n$ possess such a property. In particular, if we choose $\lambda$ according to the distribution given in \eqref{Eq:P_rho}, which is obtained by weak Fourier sampling the left coset $\ket{H}$, then $d_{\lambda}\geq n^{dn}$ with high probability, so long as $d<1/2$. The details of this statement is given in the following lemma.

\begin{lemma}\label{Lemma:P_SmallDim}
Assume $|H|\leq e^{\alpha n}$ and 
let $d<1/2$ be a constant. Let $\lambda$ be drawn from the set of irreps of $S_n$ according to the distribution $P(\cdot)$ of \eqref{Eq:P_rho}. Then for any constant $0<\gamma<1-2d$,  
\[
\pr[\lambda]{d_{\lambda}\leq n^{dn}}\leq n^{- \gamma n}
\]
for sufficiently large $n$.
\end{lemma}
\begin{proof}
As $\tr(\Pi_{H}^{\lambda})\leq d_{\lambda}$ and $|H|\leq e^{\alpha n}$, we have
\[
P(\lambda) = \frac{d_{\lambda}|H|}{|S_n|}\tr(\Pi_H^{\lambda}) \leq  \frac{d_{\lambda}^2e^{\alpha n}}{n!}\leq \frac{d_{\lambda}^2e^{\alpha n}}{n^n e^{-n}}\,.
\]
where the last inequality is obtained by applying Stirling's approximation $n! > n^n e^{-n}$.

Hence,
\[
\begin{split}
\pr[\lambda]{d_{\lambda}\leq n^{dn}}
&=\sum_{\lambda: d_{\lambda}\leq n^{dn}} P(\lambda)\\
&\leq \sum_{\lambda: d_{\lambda}\leq n^{dn}} \frac{d_{\lambda}^2e^{\alpha n}}{n^n e^{-n}}\\
&\leq p(n)n^{2dn}n^{-n}e^{\alpha n}e^{n}\\
&\leq n^{-(1-2d)n}e^{O(n)}\leq n^{-\gamma n}\\
\end{split}
\]
for sufficiently large $n$, so long as $\gamma<1-2d$.
\end{proof}
}
Now we are ready to prove the main result of this section, which is another application of Theorem~\ref{Thm:general}.

\begin{theorem}\label{Theorem:S_n} Let $H$ be a nontrivial subgroup of $S_n$ with minimal degree $m$, i.e., $m=\min_{\pi\in H\setminus\set{1}}\mathrm{supp}(\pi)$.  Then for sufficiently large $n$, $\mathcal{D}_H\leq O(|H|^2e^{-\alpha m})$.
\end{theorem}
\begin{proof}
Let $2c<d<1/2$ be constants. 
We will apply Theorem \ref{Thm:general} by setting $S=\Lambda_c$ and $D=n^{dn}$. The condition $2c<d$ guarantees that $D> d^2_S$, since $d_{S}\leq n^{cn}$ by Lemma \ref{Lemma:SmallDimS_n}.

First, we need to bound the maximal normalized character $\overline{\chi}_{\overline{S}}(H)$. By \eqref{Eq:RoichmanBound_1}, we have $\overline{\chi}_{\mu}(H)\leq e^{-\alpha m}$ for all  $\mu\in \widehat{S_n}\setminus S$. Hence, $\overline{\chi}_{\overline{S}}(H)\leq e^{-\alpha m}$. 

To bound the second term in the upper bound of Theorem \ref{Thm:general}, as $\Delta\leq |S|$, it suffices to bound:
\begin{align*}
|S|\cdot\frac{d_{S}^2}{D} 
&\leq 2cn\cdot p(cn)\cdot \frac{n^{2cn}}{n^{dn}} &\text{( by Lemma \ref{Lemma:SmallDimS_n})}\\
&\leq e^{O(\sqrt{n})}\cdot n^{(2c-d)n}  &\text{( since $cn\cdot p(cn)=e^{O(\sqrt{n})}$ )}\\
&\leq n^{-\gamma n}/2 &\text{for sufficiently large $n$, so long as $\gamma < d-2c$.}\\
\end{align*}
Now bounding the last term in the upper bound of Theorem \ref{Thm:general}:
\begin{align*}
\frac{|\overline{L_D}| D^2}{|S_n|}
&\leq \frac{p(n)n^{2dn}}{n! } & \text{(since $|\overline{L_D}|\leq |\widehat{S_n}|=p(n)$)}\\ 
&\leq \frac{e^{O(\sqrt{n})}n^{2dn}}{n^ne^{-n} } & \text{($n!> n^ne^{-n}$ by Stirling's approximation)}\\
&\leq e^{O(n)}n^{(2d-1)n}\\
&\leq n^{-\gamma n}/2 &\text{for sufficiently large $n$, so long as $\gamma < 1-2d$.}\\
\end{align*}
By Theorem \ref{Thm:general}, 
\(
\mathcal{D}_H\leq 4|H|^2(e^{-\alpha m}+n^{-\gamma n}) \,.
\)

\end{proof}

Theorem~\ref{Theorem:S_n} generalizes Moore, Russell, and Schulman's result \citep{Ref_Moore08symmetric} on strong Fourier sampling over $S_n$, which only applied in the case $|H|=2$. To relate our result to the results of \citet{Ref_Kempe07permutation} observe that, since $\log|S_n|=\Theta(n\log n)$, the subgroup $H$ is indistinguishable by strong Fourier sampling if $|H|^2e^{-\alpha m}\leq (n\log n)^{-\omega(1)}$, or equivalently, if $m\geq (2/\alpha)\log|H|+\omega(\log n)$.
\subsection{Strong Fourier sampling and the McEliece cryptosystem}
Our main application of Theorem \ref{Thm:general} is to show the limitations of strong Fourier sampling in attacking the McEliece cryptosystem. 
Throughout this section, we fix system parameters $n, k, q$ of the McEliece cryptosystem, and fix a $k\times n$  generator matrix $M$ in a private-key of the system.  
Recall that the known possible quantum attack against this McEliece cryptosystem involves solving the HSP over the wreath product group $(\GL_k(\FF_q)\times S_n)\wr \ints_2$ with the hidden subgroup being
\begin{equation}\label{Eq:SubgroupK_McEliece}
K=((H_0, s^{-1}H_0 s ),0) \cup ((H_0 s, s^{-1}H_0  ),1) 
\end{equation}
for some hidden element $s\in \GL_k(\FF_q)\times S_n$. 
Here, $H_0$ is a subgroup of $\GL_k(\FF_q)\times S_n$ given by
\begin{equation}\label{Eq:SubgroupH_0_McEliece}
H_0 = \set{(A,P)\in \GL_k(\mathbb{F}_2)\times S_n: A^{-1}MP=M}\,.
\end{equation}

Let $\aut(M)$ denote the automorphism group of the linear code generated by $M$. Observe that every element $(A,P)\in H_0$ must have $P\in\aut(M)$. This allows us to control the maximal normalized characters on $K$ through the minimal degree of $\aut(M)$. Then applying Theorem \ref{Thm:general}, we show that 

\begin{theorem}\label{Thm:McEliece} Assume $q^{k^2}\leq n^{a n}$ for some constant $0<a<1/4$. Let $m$ be the minimal degree of the automorphism group $\aut(M)$. Then for sufficiently large $n$, the subgroup $K$ defined in \eqref{Eq:SubgroupK_McEliece} has $\mathcal{D}_K \leq O(|K|^2e^{-\delta m})$ , where $\delta>0$ is a constant.
\end{theorem}

The proof of Theorem~\ref{Thm:McEliece} follows the technical ideas  discussed in the Introduction.  The details appear in Section \ref{Sec:McElieceThm}. 
 As  $q^{k^2}\leq n^{an}$, we have $\log\left|(\GL_k(\mathbb{F}_q)\times S_n)\wr\ints_2\right|=O(\log n!+\log q^{k^2})=O(n\log n)$.  Hence, the subgroup $K$ is indistinguishable if $|K|^2 e^{-\delta m} \leq (n\log n)^{-\omega(1)}$. 

In general, the size of the subgroup $K$ depends on the size of the automorphism group $\aut(M)$ and the column rank of the matrix $M$. To see this, we have $|K|=2|H_0|^2$, and $|H_0|=|\aut(M)|\times |\mathrm{Fix}(M)|$, where  $\mathrm{Fix}(M)$ denotes the set of matrices in $\GL_k(\FF_q)$ fixing $M$, i.e., 
\[
\mathrm{Fix}(M)=\set{S\in \GL_k(\FF_q)\mid SM=M}\,.
\]
To bound  the size of $\mathrm{Fix}(M)$, we record an easy fact which can be obtained by the orbit-stabilizer formula.
\begin{fact}
Let $r$ be the column rank of $M$. Then $|\mathrm{Fix}(M)|=(q^{k}-q^r)(q^k-q^{r+1})\ldots(q^k-q^{k-1})\leq q^{k(k-r)}\,.$
\end{fact}
\begin{proof}
WLOG, assume the first $r$ columns of $M$ are linearly independent, and each remaining column is a linear combination of the first $r$ columns. Consider the action of $\GL_k(\FF_q)$ on the set of $k\times n$ matrices over $\FF_q$. Under this action, the orbit of the matrix $M$, denoted $\mathrm{Orb}(M)$, consists of all $k\times m$ matrices over $\FF_q$ such that the first $r$ columns are linearly independent, and each $j^{\rm th}$ column, for $j>r$, consists of the same linear combination of the first $r$ columns as that of the  $j^{\rm th}$ column of the matrix $M$. Hence, the size of $\mathrm{Orb}(M)$ equals the number of $k\times r$ matrices over $\FF_q$ of column rank $r$. Thus, $\mathrm{Orb}(M)=(q^k-1)(q^k-q)\ldots (q^k-q^{r-1})$. On the other hand, $\mathrm{Fix}(M)$ is the stabilizer of $M$. By the orbit-stabilizer formula, we have
\[
|\mathrm{Fix}(M)|=\frac{|\GL_k(\FF_q)|}{|\mathrm{Orb}(M)|}
=\frac{(q^k-1)(q^k-q)\ldots (q^k-q^{k-1})}{(q^k-1)(q^k-q)\ldots (q^k-q^{r-1})}
=(q^{k}-q^r)(q^k-q^{r+1})\ldots(q^k-q^{k-1})\,.
\]
\end{proof}

\begin{corollary}\label{Cor:McEliece} 
Assume $q^{k^2}\leq n^{0.2n}$ and the automorphism group $\aut(M)$ has minimal degree  $\Omega(n)$. Let $r$ be the column rank of $M$. Then the subgroup $K$ defined in \eqref{Eq:SubgroupK_McEliece} has $\mathcal{D}_K \leq |\aut(M)|^4 q^{4k(k-r)} e^{-\Omega(n)}$. In particular, the subgroup $K$ is indistinguishable if, further, $|\aut(M)| \leq e^{o(n)}$ and $r\geq k- o(\sqrt{n})$.
\end{corollary}

In the a case that the matrix $M$ generates a rational Goppa code, then $M$ has full rank and the automorphism group  $\aut(M)$ is isomorphic to a subgroup of the projective linear group $\PGL_2(\FF_q)$, provided $2\leq k\leq n-2$, by Stichtenoth's Theorem \citep{Ref_Stichtenoth90on} (see Appendix \ref{ApxSec:Goppa} for more detailed background on rational Goppa codes). This important property results in very good values we could desire for the automorphism group $\aut{(M)}$: we have $|\aut(M)|\leq |\PGL_2(\FF_q)|\leq q^3$, and moreover,
\begin{lemma}\label{Lemma:MinSupportAut}
If $M$ generates a rational Goppa code, the minimal degree of $\aut(M)$ is at least $n-3$.
\end{lemma}
\begin{proof} (sketch)
Since $\aut(M)$ is isomorphic to a subgroup of $\PGL_k(\mathbb{F}_q)$, 
the proof is based on the observation that  any transformation in $\PGL_k(\mathbb{F}_q)$ that fixes at least three distinct projective lines must be the identity.
\end{proof}

By Corollary \ref{Cor:McEliece}, the McEliece crytosystem coupled with rational Goppa codes resists known quantum attacks based on strong Fourier sampling. Unfortunately, this cryptosystem is insecure due to the classical attack by Sidelnokov and Shestakov \cite{Ref_Sidelnikov92insecurity} that recovers the secret scrambler $A$ and the product $MP$, but does not reveal the secret permutation $P$. Of course, our next goal will be to find other classes of linear codes with which the McEliece cryptosystem would be secure against both classical and quantum attacks.

\subsection{Strong Fourier sampling over  $\GL_2(\mathbb{F}_q)$}
In this simple application, we consider the finite general linear group $G=\GL_2(\mathbb{F}_q)$, whose structure as well as irreps are well established \citep[\S 5.2]{Ref_Fulton91representation}. 
From the character table of $\GL_{2}(\mathbb{F}_q)$, which can be found in Appendix \ref{ApxSec:GL_2}, we draw the following easy facts:

\begin{fact}
Let $\sigma$ be an irrep of $\GL_{2}(\mathbb{F}_q)$. Then
\textit{(i)} For all $g\in \GL_{2}(\mathbb{F}_q)$, $|\chi_{\sigma}(g)|=d_{\sigma}$ if $g$ is a scalar matrix, and $|\chi_{\sigma}(g)|\leq 2$ otherwise. 
\textit{(ii)} If $d_{\sigma}>1$, then $q-1\leq d_{\sigma}\leq q+1$.
\end{fact}
\remove{
\begin{fact}
There are $q-1$ linear representations of $\GL(2,q)$ (corresponding to the $q-1$ characters of the cyclic group $\mathbb{F}^*_q$), 
\end{fact}
}
Let $H$ be a subgroup of $\GL_2(\mathbb{F}_q)$. If $H$ contains a non-identity scalar matrix, we have $\overline{\chi}_{\sigma}(H)=1$ for all $\sigma$, making it impossible to find a set of irreps whose maximal normalized characters on $H$ are small enough to apply our general theorem (Theorem \ref{Thm:general}). For this reason, we shall assume that $H$ does not contain scalar matrices except for the identity. 
An example of such a subgroup $H$ is any group lying inside the subgroup of triangular unipotent matrices
\(
\set{
T(b) \mid b\in \mathbb{F}_q }
\), 
where 
$T(b):=\begin{pmatrix}
  1    &  b   \\
  0    &  1
\end{pmatrix}$.

From the easy facts above for $\GL(2,q)$, it is natural to choose the set $S$ in Theorem \ref{Thm:general} to be the set of linear (i.e., $1$-dimensional) representations, and choose the dimensional threshold $D$ to be $q-1$. However, since $\GL(2,q)$ has $q-1$ linear representations,  
i.e., $|S|=D$, we can't upper bound $\Delta$ by $|S|$. We prove the following lemma to provide a strong upper bound on $\Delta$, which is, in this case, the maximal number of linear representations appearing in the decomposition of $\rho\otimes\rho^*$, for any nonlinear irrep $\rho$.

\begin{lemma}\label{Lemma:LinearRep_0}
Let $\rho$ be an irrep of $\GL(2,q)$. Then at most two linear representations appear in the decomposition of $\rho\otimes \rho^*$. 
\end{lemma}

The proof for this lemma can be found in Appendix \ref{ApxSec:GL_2}. Then applying Theorem \ref{Thm:general} with $S$ being the set of linear representations, and $L$ being the set of non-linear irreps of $\GL_2(\mathbb{F}_q)$, we have: 

\begin{corollary}\label{Theorem:GL_2_q}
Let $H$ be a subgroup of $\GL_2(\FF_q)$ that does not contain any scalar matrix other than the identity.  Then $\mathcal{D}_H\leq {28|H|^2}/{q}$.
\end{corollary}
\begin{proof}[Proof of Corollary \ref{Theorem:GL_2_q}]
Let $S$ be the set of linear representations of $\GL_2(\mathbb{F}_q)$ and let $D=q-1$. Then in this case, $L_D$ is the set of all non-linear irreps of $\GL_2(\mathbb{F}_q)$.

Since $\overline{\chi}_{\sigma}(H)\leq 2/(q-1)$ for all nonlinear irrep $\sigma$, we have 
$$\overline{\chi}_{\overline{S}}(H)\leq 2/(q-1)\leq 0.5/|H| \,.$$

To bound the second term in the bound of \ref{Thm:general}, we have $\Delta \leq 2$ by Lemma \ref{Lemma:LinearRep_0} and $d_S=1$, thus
\[
\Delta\frac{d_S^2}{D} \leq 2/(q-1) \leq 3/q\,.
\]
As $|G|=(q-1)^2q(q+1)$ and $|\overline{L_D}|=|S|=q-1$, we have
\[
\frac{|\overline{L_D}| D^2}{|G|} =\frac{(q-1)^3}{(q-1)^2q(q+1)} = \frac{q-1}{q(q+1)} < 1/q\,.
\] 
By Theorem \ref{Thm:general}, \( \mathcal{D}_H\leq 4|H|^2\left(7/q\right)\,.
\)
\end{proof}

In particular, $H$ is indistinguishable by strong Fourier sampling over $\GL_2(\FF_q)$ if $|H|\leq q^{\delta}$ for some $\delta<1/2$, because in that case we have
\(
\mathcal{D}_H\leq 28q^{2\delta -1} \leq \log^{-c}|\GL_2(\FF_q)|
\)
for all constant $c>0$.

\paragraph{Examples of indistinguishable subgroups.} As a specific example, consider a cyclic subgroup $H_b$  generated by a triangular unipotent matrix $T(b)$
\remove{
$\begin{pmatrix}
  1    &  b   \\
  0    &  1
\end{pmatrix}$,} for any $b\neq 0$. 
Since $T(b)^k = T(kb)$
\remove{
$\begin{pmatrix}
  1    &  b   \\
  0    &  1
\end{pmatrix}^k = \begin{pmatrix}
  1    &  kb   \\
  0    &  1
\end{pmatrix}$} for any integer $k\geq 0$, the order of $H_b$ is the least positive integer $k$ such that $kb=0$. In particular, the order of $H_b$ equals the characteristic of the finite field $\mathbb{F}_q$.
Suppose $q=p^n$ for some prime number $p$ and $n>2$. Then $\mathbb{F}_q$ has characteristic $p$, and hence, $|H_b|=p$. By Corollary \ref{Theorem:GL_2_q}, we have $\mathcal{D}_{H_b}\leq O(p^{2-n})$. 

Similarly, consider a subgroup $H_{a,b}$ generated by two distinct non-identity elements $T(a)$ and $T(b)$.
\remove{ 
\(
\begin{pmatrix}
  1    &  a   \\
  0    &  1
\end{pmatrix} ~ \text{and}~ \begin{pmatrix}
  1    &  b   \\
  0    &  1
\end{pmatrix}\,.
\)}
Since elements of $H_{a,b}$ are of the form $T(ka+\ell b)$ 
\remove{
$\begin{pmatrix}
  1    &  ka+\ell b   \\
  0    &  1
\end{pmatrix}$
}
for $k,\ell\in\set{0,1,\ldots,p-1}$, we have $|H_{a,b}|\leq p^2$. 
Thus, the distinguishability of $H_{a,b}$ using strong Fourier sampling over $\GL_2(\FF_{p^n})$ is $O(p^{4-n})$. Clearly, both $H_b$ and $H_{a,b}$ are indistinguishable, for $n$ sufficiently large.  More generally, any subgroup generated by a constant number of triangular unipotent matrices is indistinguishable.  


\section{Bounding distinguishability}\label{Sec:MainTheoremProof}
We  now present the proof for the main theorem (Theorem \ref{Thm:general}) in details. 
Fixing a nontrivial subgroup $H<G$, we want to upper bound $\mathcal{D}_H$.
Let us start with bounding the expectation over the random group element $g\in G$, for a fixed irrep $\rho\in\widehat{G}$:  
\[
E_H(\rho) \eqdef \expect[g]{\|P_{H^g}(\cdot\mid\rho)-U_{B_{\rho}}\|_{1}^2}\,.
\]
Obviously we always have $E_H(\rho) \leq 4$. More interestingly, we have
\begin{eqnarray}
E_H(\rho) 
&=&\expect[g]{\left(\sum_{\mathbf{b}\in B_{\rho}} \left|P_{H^g}(\mathbf{b}\mid\rho)-\frac{1}{d_{\rho}}\right|\right)^2}  \nonumber\\
&\leq &  \expect[g]{d_{\rho}\sum_{\mathbf{b}\in B_{\rho}}\left(P_{H^g}(\mathbf{b}\mid\rho)-\frac{1}{d_{\rho}}\right)^2} 
\quad\text{(by Cauchy-Schwarz)}\nonumber\\
&=&   d_{\rho}\sum_{\mathbf{b}\in B_{\rho}} \var[g]{P_{H^g}(\mathbf{b}\mid\rho)} 
\qquad\qquad\quad\text{(since $\expect[g]{P_{H^g}(\mathbf{b}\mid\rho)}=\frac{1}{d_{\rho}}$)} \nonumber\\
&=& \frac{d_{\rho}}{\tr(\Pi_H^{\rho})^2}\sum_{\mathbf{b}\in B_{\rho}} \var[g]{\|\Pi_{H^g}^{\rho}\mathbf{b}\|^2}\label{Eq:E_rho_1}\,. 
\end{eqnarray}
\noindent
The equation $\expect[g]{P_{H^g}(\mathbf{b}\mid\rho)}={1}/{d_{\rho}}$ (Proposition \ref{Prop:ExpSchur} in Appendix \ref{ApxSec:MainThm}) can be shown using \emph{Schur's lemma}.

From \eqref{Eq:E_rho_1}, we are motivated to bound the variance of $\|\Pi_{H^g}^{\rho}\mathbf{b}\|^2$ when $g$ is chosen uniformly at random. We provide an upper bound that depends on the projection of the vector $\mathbf{b}\otimes\mathbf{b}^*$ onto irreducible subspaces of $\rho\otimes\rho^*$, and on maximal normalized characters of $\sigma$ on $H$ for all irreps $\sigma$ appearing in the decomposition of $\rho\otimes\rho^*$. 
Recall that the representation $\rho\otimes\rho^*$ is typically reducible and can be written as an orthogonal direct sum of irreps $\rho\otimes\rho^*=\bigoplus_{\sigma\in\widehat{G}} a_{\sigma}\sigma$, where $a_{\sigma}\geq 0$ is the multiplicity of $\sigma$. Then $I(\rho\otimes \rho^*)$ consists of $\sigma$  with $a_{\sigma}>0$, and 
we let $\Pi_{\sigma}^{\rho\otimes \rho^*}$ denote the projection operator whose image is $a_{\sigma}\sigma$, that is, the subspace spanned by all copies of $\sigma$.
Our upper bound given in Lemma \ref{Lemma:QFS-var} below generalizes the bound given in Lemma 4.3 of \citep{Ref_Moore08symmetric}, which only applies to subgroups $H$ of order 2. 
\begin{lemma}\label{Lemma:QFS-var} Let $\rho$ be an irrep of $G$. Then for any vector $\mathbf{b}\in V_{\rho}$, 
\[
\var[g]{\|\Pi_{H^g}^{\rho}\mathbf{b}\|^2} 
\leq 
\sum_{\sigma \in I(\rho\otimes \rho^*)}\overline{\chi}_{\sigma}(H) \left\| \Pi_{\sigma}^{\rho\otimes \rho^*}(\mathbf{b}\otimes\mathbf{b}^*)\right\|^2\,.
\]
\end{lemma}
\begin{proof}[Proof of Lemma~\ref{Lemma:QFS-var}] Fix a vector $\mathbf{b}\in V_{\rho}$. To simplify notations, we shall write $\Pi_{g}$ as shorthand for $\Pi_{H^g}^{\rho}$, and write $g\mathbf{b}$ for $\rho(g)\mathbf{b}$. 
For any $g\in G$, we have
\[
\begin{split}
\|\Pi_{g}\mathbf{b}\|^2 
&=\tup{\Pi_{g}\mathbf{b}, \Pi_{g}\mathbf{b}}
=\tup{\mathbf{b},\Pi_{g}\mathbf{b}}\\
&= \frac{1}{|H|}\left(\tup{\mathbf{b},\mathbf{b}}+\sum_{h\in H\setminus\set{1}}\tup{\mathbf{b},g^{-1}hg\mathbf{b}}\right)\,.\\
\end{split}
\]
Let $S_g=\sum_{h\in H\setminus\set{1}}\tup{\mathbf{b},g^{-1}hg\mathbf{b}}$. Then
\[
\|\Pi_{g}\mathbf{b}\|^2 =\frac{1}{|H|}\left(\|\mathbf{b}\|^2+S_g\right)\quad \text{and}\quad
S_g = |H|\|\Pi_{g}\mathbf{b}\|^2  - \|\mathbf{b}\|^2\,.
\]
It follows that $S_g$ is real, and that
\[
\begin{split}
\|\Pi_{g}\mathbf{b}\|^4 
&= \frac{1}{|H|^2}\left(\|\mathbf{b}\|^4+2\|\mathbf{b}\|^2S_g+ S_g^2\right)\,.\\
\end{split}
\]
We have
\begin{align}
\expect[g]{\|\Pi_{g}\mathbf{b}\|^4} 
&= \frac{1}{|H|^2}\left(\|\mathbf{b}\|^4+2\|\mathbf{b}\|^2\expect[g]{S_g}+ \expect[g]{S_g^2}\right)\label{Eq:E1}\\ 
\expect[g]{\|\Pi_{g}\mathbf{b}\|^2}^2 
&=  \frac{1}{|H|^2}\left(\|\mathbf{b}\|^2+\expect[g]{S_g}\right)^2 \nonumber\\
&=  \frac{1}{|H|^2}\left(\|\mathbf{b}\|^4+2\|\mathbf{b}\|^2\expect[g]{S_g}+ \expect[g]{S_g}^2\right)\label{Eq:E2}
\end{align}
Subtracting \eqref{Eq:E1} by \eqref{Eq:E2} yields
\[
\begin{split}
\var[g]{\|\Pi_{g}\mathbf{b}\|^2} 
&=\expect[g]{\|\Pi_{g}\mathbf{b}\|^4}-\expect{\|\Pi_{g}\mathbf{b}\|^2}^2\\
&= \frac{1}{|H|^2}\left(\expect[g]{S_g^2}-\expect[g]{S_g}^2\right)
\leq \frac{1}{|H|^2}\expect[g]{S_g^2}\,.\\
\end{split}
\]
To bound the variance, we upper bound $S_g^2$ for all $g\in G$. 
Since $S_g$ is real, applying Cauchy-Schwarz inequality, we have
\[
S_g^2
=\left|\sum_{h\in H\setminus\set{1}}\tup{\mathbf{b},g^{-1}hg\mathbf{b}}\right|^2 
\leq (|H|-1)\left(\sum_{h\in H\setminus\set{1}}\left|\tup{\mathbf{b},g^{-1}hg\mathbf{b}}\right|^2\right)\,.
\]

Proving similarly to Lemma 4.2 in \citep{Ref_Moore08symmetric}, one can express the second moment of the inner product $\tup{\mathbf{b},g^{-1}hg\mathbf{b}}$ in terms of the projection of $\mathbf{b}\otimes\mathbf{b}^*$ into the irreducible constituents of the tensor product representation $\rho\otimes \rho^*$. Specifically, for any $h\in G$, we have
\[
\expect[g]{|\tup{\mathbf{b},g^{-1}hg\mathbf{b}}|^2} 
= \sum_{\sigma \in I(\rho\otimes \rho^*)}\frac{\chi_{\sigma}(h)}{d_{\sigma}}\left\| \Pi_{\sigma}^{\rho\otimes \rho^*}(\mathbf{b}\otimes\mathbf{b}^*)\right\|^2\,.
\]

It follows that 
\[
\var[g]{\|\Pi_{H^g}^{\rho}\mathbf{b}\|^2} 
\leq 
\frac{1}{|H|}\left(\sum_{h\in H\setminus\set{1}}\expect[g]{\left|\tup{\mathbf{b},g^{-1}hg\mathbf{b}}\right|^2}\right)
\leq 
\sum_{\sigma \in I(\rho\otimes \rho^*)}\overline{\chi}_{\sigma}(H) \left\| \Pi_{\sigma}^{\rho\otimes \rho^*}(\mathbf{b}\otimes\mathbf{b}^*)\right\|^2\,.
\]
\end{proof}


Back to our goal of bounding $E_H(\rho)$ using the bound in Lemma \ref{Lemma:QFS-var}, the strategy will be to separate irreps appearing in the decomposition of $\rho\otimes \rho^*$ into two groups, those with small dimension and those with large dimension, and treat them differently. If $d_{\sigma}$ is large, we shall rely on bounding $\overline{\chi}_{\sigma}(H)$. If $d_{\sigma}$ is small, we shall control the projection given by $\Pi_{\sigma}^{\rho\otimes\rho^*}$ using 
the following lemma which was proved implicitly in \citep{Ref_Moore08symmetric} (its proof is also given in Appendix): 

\begin{lemma}\label{Lemma:LargeSmall} 
For any irrep $\sigma$, we have
\(
\sum_{\mathbf{b}\in B_{\rho}}\left\|\Pi_{\sigma}^{\rho\otimes\rho^*}(\mathbf{b}\otimes \mathbf{b}^*)\right\|^2 \leq d_{\sigma}^2\,.
\)
\end{lemma}

The method discussed above for bounding $E_H(\rho)$ is culminated into Lemma \ref{Lemma:general_method} below.

\begin{lemma}\label{Lemma:general_method}
Let $\rho\in\widehat{G}$ be arbitrary and $S\subset \widehat{G}$ be any subset of irreps that does not contain $\rho$.  
Then
\[
E_H(\rho) \leq 4|H|^2\left(\overline{\chi}_{\overline{S}}(H) + \left|S\cap I(\rho\otimes \rho^*)\right|\frac{d_S^2}{d_{\rho}}\right)\,.
\]
\end{lemma}
\begin{proof}[Proof of Lemma \ref{Lemma:general_method}]

Combining Inequality \eqref{Eq:E_rho_1} and Lemmas \ref{Lemma:QFS-var}  give
\[
\begin{split}
E_H(\rho) \leq \frac{d_{\rho}}{\tr(\Pi_H^{\rho})^2}\sum_{\sigma \in I(\rho\otimes \rho^*)}\overline{\chi}_{\sigma}(H)\sum_{\mathbf{b}\in B_{\rho}}
\left\| \Pi_{\sigma}^{\rho\otimes \rho^*}(\mathbf{b}\otimes\mathbf{b}^*)\right\|^2\,.
\end{split}
\]
Now we split additive items in the above upper bound into two groups separated by the set $S$. 
For the first group (large dimension),
\[
\begin{split}
\sum_{\sigma \in \overline{S}\cap\widehat{G}^{\rho\otimes \rho^*}}\overline{\chi}_{\sigma}(H)\sum_{\mathbf{b}\in B_{\rho}}
\left\| \Pi_{\sigma}^{\rho\otimes \rho^*}(\mathbf{b}\otimes\mathbf{b}^*)\right\|^2 
&\leq 
\overline{\chi}_{\overline{S}}(H)\sum_{\mathbf{b}\in B_{\rho}}\underbrace{\sum_{\sigma\in I(\rho\otimes \rho^*) }
\left\| \Pi_{\sigma}^{\rho\otimes \rho^*}(\mathbf{b}\otimes\mathbf{b}^*)\right\|^2}_{\leq 1} \\
&\leq \overline{\chi}_{\overline{S}}(H)d_{\rho}\,.
\end{split}
\]
For the second group (small dimension), 
\[
\begin{split}
\sum_{\sigma \in S\cap I(\rho\otimes \rho^*)}\overline{\chi}_{\sigma}(H)\sum_{\mathbf{b}\in B_{\rho}}
\left\| \Pi_{\sigma}^{\rho\otimes \rho^*}(\mathbf{b}\otimes\mathbf{b}^*)\right\|^2
&\leq 
\sum_{\sigma \in S\cap I(\rho\otimes \rho^*)}\overline{\chi}_{\sigma}(H)d_{\sigma}^2  
\qquad\text{(by Lemma \ref{Lemma:LargeSmall})}\\
&\leq 
\sum_{\sigma \in S\cap I(\rho\otimes \rho^*)}d_{\sigma}^2
\qquad\qquad\text{(since $\overline{\chi}_{\sigma}(H)\leq 1$)}\\
&\leq \left|S\cap I(\rho\otimes \rho^*)\right|d_S^2\,.
\end{split}
\]
Summing the last bounds for the two groups yields
\[
E_H(\rho) \leq \left(\frac{d_{\rho}}{\tr(\Pi_H^{\rho})} \right)^2
\left( \overline{\chi}_{\overline{S}}(H)+\left|S\cap I(\rho\otimes \rho^*)\right|\frac{d_S^2}{d_{\rho}}\right)\,.
\]
On the other hand, since $E_H(\rho) \leq 4$, we can assume $H^2\overline{\chi}_{\overline{S}}(H)\leq 1$, and thus $\overline{\chi}_{\overline{S}}(H)\leq \frac{1}{|H|^2}\leq \frac{1}{2|H|}$.
Hence, we have
\[
\frac{\tr(\Pi_H^{\rho})}{d_{\rho}} 
=\frac{1}{|H|}\left(1+\sum_{h\in H\setminus\set{1}}\frac{\chi_{\rho}(h)}{d_{\rho}}\right)
\geq \frac{1}{|H|}-\overline{\chi}_{\rho}(H)
\geq \frac{1}{2|H|}\,,
\]
where the last inequality is due to $\overline{\chi}_{\rho}(H)\leq \overline{\chi}_{\overline{S}}(H)\leq \frac{1}{2|H|}$.
This completes the proof.
\end{proof}

To apply this lemma, we should choose the subset $S$ such that $d^2_S \ll d_{\rho}$, that is, $S$ should consist of small dimensional irreps. 
Then applying Lemma \ref{Lemma:general_method} for all irreps $\rho$ of large dimension, we can prove our general main theorem straightforwardly.

\begin{proof}[\textit{Proof of Theorem \ref{Thm:general}:}]
For any $\rho\in L$, since $d_{\rho}\geq D> d_S^2 $, we must have $\rho\not\in S$. 
By Lemma \ref{Lemma:general_method}, 
\[
E_H(\rho) \leq 4|H|^2\left(\overline{\chi}_{\overline{S}}(H)+\Delta\frac{d_S^2}{D}\right)\quad\text{for all $\rho\in L$}\,.
\]
Combining this with the fact that $E_H(\rho)\leq 4$ for all $\rho\not\in L$, we obtain
\[
\begin{split}
\mathcal{D}_H 
&=\expect[\rho]{E_H(\rho)}
\leq 4|H|^2\left(\overline{\chi}_{\overline{S}}(H)+\Delta\frac{d_S^2}{D}\right)+4\pr[\rho]{\rho\not\in L}\,.\\
\end{split}
\]
To complete the proof, it remains to bound $\pr[\rho]{\rho\not\in L}$. 
Since $\tr(\Pi_H^{\rho})\leq d_{\rho}$, we have
\[
P(\rho) = \frac{d_{\rho}|H|}{|G|}\tr(\Pi_H^{\rho}) \leq \frac{d_{\rho}^2|H|}{|G|}\,.
\]
Since $d_{\rho}< D$ for all $\rho\in\widehat{G}\setminus L$, it follows that
\[
\pr[\rho]{\rho\not\in L} = \sum_{\rho\not\in L}P(\rho)\leq \frac{|\overline{L}|D^2|H|}{|G|} \leq \frac{|\overline{L}|D^2|H|^2}{|G|}\,.
\]
\remove{
Finally, apply the following Markov's inequality to get the stated result:
\[
\pr[\rho,g]{\|P_g(\rho,\cdot)-U_{B_{\rho}}\|_{1}\geq \epsilon} \leq 
\expect[\rho,g]{\|P_{g}(\rho,\cdot)-U_{B_{\rho}}\|_{1}^2}/\epsilon^2 \,.
\]
}
\end{proof}

\section{Strong Fourier sampling over  $(\GL_k(\FF_q)\times S_n)\wr \ints_2$}\label{Sec:McElieceThm}
This section devotes to the proof of Theorem \ref{Thm:McEliece} which establishes the limitation of strong Fourier sampling  in breaking the McEliece cryptosystem. The goal is to bound the distinguishability of the subgroup $K$ defined in \eqref{Eq:SubgroupK_McEliece} of the wreath product  $(\GL_k(\FF_q)\times S_n)\wr \ints_2$.


\subsection{Normalized characters for $G\wr\ints_2$}

Firstly, we  consider quantum Fourier sampling over the wreath product $G\wr\ints_2$, for a general group $G$, with a hidden subgroup of the form
\[
K=((H_0, s^{-1}H_0 s ),0) \cup ((H_0 s, s^{-1}H_0  ),1) < G\wr\ints_2
\]
for some subgroup $H_0< G$ and some element $s\in G$. Again, the first thing we need to understand is the maximal normalized characters on $K$. Recall that all irreducible characters of $G\wr\ints_2$ are constructed in the following ways:
\begin{enumerate}
\item Each unordered pair of two non-isomorphic irreps $\sigma, \rho\in\widehat{G}$ gives rise to an irrep of $G\wr\ints_2$, denoted ${\{\rho, \sigma\}}$,  with character given by:
\[
\chi_{\set{\rho,\sigma}}((x,y),b) = 
\begin{cases}
\chi_{\rho}(x)\chi_{\sigma}(y)+\chi_{\rho}(y)\chi_{\sigma}(x) &\text{if } b=0\\
0 &\text{if } b=1\,.\\
\end{cases}
\]
The dimension of representation ${\{\rho, \sigma\}}$ is equal to $\chi_{\set{\rho,\sigma}}((1,1),0)=2d_{\rho}d_{\sigma}$.
\item Each irrep $\rho\in\widehat{G}$  gives rise to two irreps  of $G\wr\ints_2$, denoted ${\{\rho\}}$ and ${\{\rho\}'}$, with characters given by:
\[
\chi_{\set{\rho}}((x,y),b) =
\begin{cases}
\chi_{\rho}(x)\chi_{\rho}(y) &\text{if } b=0\\
\chi_{\rho}(xy) &\text{if } b=1\\
\end{cases}
\]
\[
\chi_{\set{\rho}'}((x,y),b) =
\begin{cases}
\chi_{\rho}(x)\chi_{\rho}(y) &\text{if } b=0\\
-\chi_{\rho}(xy) &\text{if } b=1\,.\\
\end{cases}
\]
Both representations ${\{\rho\}}$ and ${\{\rho\}'}$ have the same dimension equal $d_{\rho}^2$.
\end{enumerate}

Clearly, the number of irreps of $G\wr\ints$ is equal to $|\widehat{G}|^2/2+3|\widehat{G}|/2$, which is no more than $|\widehat{G}|^2$ as long as $G$ has at least three irreps.
Now it is easy to determine the maximal normalized characters on subgroup $K$.
\begin{proposition}
For non-isormorphic irreps $\rho, \sigma\in\widehat{G}$,
\[
\overline{\chi}_{\set{\rho,\sigma}}(K) \leq \overline{\chi}_{\rho}(H_0)\overline{\chi}_{\sigma}(H_0)\,.
\]
For irrep $\rho\in\widehat{G}$,
\[
\overline{\chi}_{\set{\rho}}(K) = \overline{\chi}_{\set{\rho}'}(K) =\max\set{\overline{\chi}_{\rho}(H_0)^2, 1/d_{\rho}} \,.
\]
\end{proposition}

So to bound the maximal normalized characters over $K$, we can turn to bounding the normalized characters on  the subgroup $H_0$ and the dimension of an irrep of $G$.

\subsection{Normalized characters for $(\GL_k(\mathbb{F}_q)\times S_n)\wr\ints_2$}
Recall that for the case of attacking McEliece cryptosystem, we have $G=\GL_k(\mathbb{F}_q)\times S_n$ and every element $(A,P)\in H_0$ has $P\in\aut(M)$.

For $\tau\in \widehat{\GL_k(\mathbb{F}_q)}$ and $\lambda\in\widehat{S_n}$, let $\tau\times \lambda$ denote the tensor product as a representation of $\GL_k(\mathbb{F}_q)\times S_n$. 
Those tensor product representations  $\tau\times \lambda$ are all irreps of $\GL_k(\mathbb{F}_q)\times S_n$. Since $\overline{\chi}_{\tau\times \lambda}(S_{\pi}, \pi)=\overline{\chi}_{\tau}(S_{\pi})\overline{\chi}_{\lambda}(\pi)$ and $\overline{\chi}_{\tau}(S_{\pi})\leq 1$ for all $\pi\in S_n$, we have
\[
\overline{\chi}_{\tau\times \lambda}(H_0) \leq \overline{\chi}_{\lambda}(\aut(M))\,.
\]

As in the treatment for the symmetric group, we can bound the maximal normalized character $\overline{\chi}_{\lambda}(\aut(M))$ based on the minimum support of non-identity elements in $\aut(M)$, for any $\lambda\in\widehat{S_n}\setminus\Lambda_c$.

To complete bounding the maximal normalized characters on the subgroup $K$, it remains to bound the dimension of a representation ${\tau\times\lambda}$ of the group $\GL_k(\mathbb{F}_q)\times S_n$ with $\lambda\in \widehat{S_n}\setminus \Lambda_c$. Since the dimension of ${\tau\times\lambda}$ is  $$d_{\tau\times\lambda}=d_{\tau}d_{\lambda}\geq d_{\lambda}\,,$$ we prove the following lower bound for $d_{\lambda}$. 

\begin{lemma}\label{Lemma:LargeDimS_n}
Let $0<c\leq 1/6$ be a constant. Then there is a constant $\beta>0$ depending only on $c$ such that for sufficiently large $n$ and for $\lambda\in \widehat{S_n}\setminus \Lambda_c$, 
\[
d_{\lambda} \geq e^{\beta n}\,.
\]
\end{lemma}
\begin{proof}[Proof of Lemma~\ref{Lemma:LargeDimS_n}]
Consider an integer partition of $n$, $\lambda=(\lambda_1,\ldots,\lambda_t)$, with both $\lambda_1$ and $t$ less than $(1-c)n$. Let $\lambda'=(\lambda'_1,\ldots,\lambda'_{\lambda_1})$ be the conjugate of $\lambda$, where $t=\lambda'_1\geq \lambda'_2 \geq \ldots \geq \lambda'_{\lambda_1}$ and $\sum_{i}\lambda'_i=n$. WLOG, assume $\lambda'_1\leq \lambda_1$. We label all the cells of the Young diagram of shape $\lambda$ as $c_1,\ldots,c_n$, in which $c_i$ is the $i^{\rm th}$ cell from the left of the first row, for $1\leq i\leq \lambda_1$.

The dimension of $\lambda$ is determined by the \emph{hook length formula}:
\[
d_{\lambda} = \frac{n!}{\mathrm{Hook}(\lambda)}\,,\quad\quad \mathrm{Hook}(\lambda) = \prod_{i=1}^n \mathrm{hook}(c_i)\,,
\]
where $\mathrm{hook}(c_i)$ is the number of cells appearing in either the same column or the same row as the cell $c_i$, excluding those that are above or the the left of $c_i$. In particular,
\[
\mathrm{hook}(c_i) = \lambda_1-i+\lambda'_i \quad\quad\text{for $1\leq i\leq \lambda_1$.}
\] 

If $\lambda_1\leq cn$, we have $\mathrm{hook}(c_i)\leq t+\lambda_1\leq 2cn$ for all $i$, thus
\[
d_{\lambda} \geq \frac{n!}{(2cn)^n} \geq \frac{n^n}{e^n(2cn)^n} \geq \left(\frac{3}{e}\right)^{n}\geq e^{\beta n}\,.
\]

Now we consider the case  $cn< \lambda_1 < (1-c)n$.
Let $\tilde{\lambda}=(\lambda_2,\ldots,\lambda_t)$, this is an integer partition of $n-\lambda_1$ whose Young diagram is obtained by removing the first row of $\lambda$. Applying the hook length formula for $\tilde{\lambda}$ and the fact that $d_{\tilde{\lambda}}\geq 1$ gives us:
\[
Hook(\tilde{\lambda}) =\frac{(n-\lambda_1)!}{d_{\tilde{\lambda}}} \leq (n-\lambda_1)!\,.
\]
Then we have
\[
\mathrm{Hook}(\lambda) = \mathrm{Hook}(\tilde{\lambda}) \prod_{i=1}^{\lambda_1}\mathrm{hook}(c_i)\leq (n-\lambda_1)!\prod_{i=1}^{\lambda_1}\mathrm{hook}(c_i)\,. 
\]
On the other hand, we have
\[
\begin{split}
\prod_{i=1}^{\lambda_1}\mathrm{hook}(c_i)
&= \prod_{i=1}^{\lambda_1} (\lambda_1-i+\lambda'_i)\\
&= \lambda_1!\prod_{i=1}^{\lambda_1} \left(1+\frac{\lambda'_i-1}{\lambda_1-i+1}\right)\\
&\leq \lambda_1!\exp\left(\sum_{i=1}^{\lambda_1}\frac{\lambda'_i-1}{\lambda_1-i+1}\right) \quad\text{(since $1+x\leq e^x$ for all $x$).}
\end{split}
\]
To upper bound the exponent in the last equation, we use Chebyshev's sum inequality, which states that for any increasing sequence $a_1\geq a_2\geq \ldots \geq a_k$ and any decreasing sequence $b_1\leq b_2\leq \ldots \leq b_k$ or real numbers, we have $k\sum_{i=1}^k a_ib_i \leq \left(\sum_{i=1}^k a_i\right)\left(\sum_{i=1}^k b_i\right)$. Since the sequence $\set{\lambda'_i-1}$ is increasing and the sequence $\set{1/(\lambda_1-i+1)}$ is decreasing, we get
\[
\begin{split}
\sum_{i=1}^{\lambda_1}\frac{\lambda'_i-1}{\lambda_1-i+1} 
&\leq \frac{\sum_{i=1}^{\lambda_1}(\lambda'_i-1)}{\lambda_1}\left(\sum_{i=1}^{\lambda_1}\frac{1}{\lambda_1-i+1}\right)\\
&= \frac{n-\lambda_1}{\lambda_1}\left(\sum_{i=1}^{\lambda_1}\frac{1}{i}\right)
\leq \frac{1}{c}\left(\sum_{i=1}^{\lambda_1}\frac{1}{i}\right) \quad\quad\text{(since $\lambda_1>cn$)}\,.
\end{split}
\]
Let $r$ be a constant such that $1 <r/c <cn$. Bounding $1/i\leq 1$ for all $i\leq r/c$ and bounding $1/i\leq c/r$ for all $i>r/c$ yields
\[
\sum_{i=1}^{\lambda_1}\frac{1}{i} \leq \frac{r}{c}+ \frac{c\lambda_1}{r}\,.  
\]
Putting the pieces together, we get
\begin{align*}
d_{\lambda} 
&\geq  \frac{n!}{(n-\lambda_1)!\lambda_1! e^{\lambda_1/r+r/c^2}}  
={n\choose \lambda_1}e^{-\lambda_1/r-r/c^2} & \\ 
&\geq \left(\frac{n}{\lambda_1}\right)^{\lambda_1} e^{-\lambda_1/r-r/c^2} & \\ 
&\geq \left(\frac{e^{-1/r}}{1-c}\right)^{\lambda_1} e^{-r/c^2}&\text{(since $ \lambda_1 < (1-c)n$)}\,. &  
\end{align*}
Let $0<\delta < \ln\frac{1}{1-c}$ be a constant and choose $r$ large enough so that $e^{-1/r}\geq (1-c)e^{\delta}$. Then
\[
d_{\lambda}\geq e^{\delta\lambda_1-r/c^2}  \geq e^{\delta cn-r/c^2} \geq e^{\beta n}\,.
\]
\end{proof}

\begin{remark}
The lower bound in Lemma \ref{Lemma:LargeDimS_n} is essentially tight. To see this, consider the hook of width $(1-c)n$ and of depth $cn$. This hook has dimension roughly equal ${n\choose cn}$, which is no more than $(e/c)^{cn}$.
\end{remark}

\subsection{Proof of Theorem \ref{Thm:McEliece}}

We are ready to prove Theorem \ref{Thm:McEliece}. 

\begin{proof}[Proof of Theorem \ref{Thm:McEliece}]
To apply Theorem \ref{Thm:general}, let $0<c<\min\set{1/6, 1/4-a}$ be a constant and $S$ be the set of irreps of $(\GL_k(\mathbb{F}_q)\times S_n)\wr\ints_2$ of the forms $\set{\tau\times\lambda, \eta\times \mu}$, $\set{\tau\times \lambda}$, $\set{\tau\times \lambda}'$ with $\tau,\eta\in \widehat{\GL_k(\FF_q)}$ and $\lambda,\mu\in\Lambda_c$, where  $\Lambda_c$ is mentioned in Section \ref{SubSec:Symmetric}. Firstly, we need upper bounds for $\overline{\chi}_{\overline{S}}(K)$, $|S|$, and $d_S$.

Since $\aut(M)$ has minimal degree  $m$, by Inequality \eqref{Eq:RoichmanBound_1} in Section \ref{SubSec:Symmetric},  we have for all $\lambda\in \widehat{S_n}\setminus \Lambda_c$,
\[
\overline{\chi}_{\lambda}(\aut(M)) \leq e^{-\alpha m} \,.
\]
Combining with Lemma \ref{Lemma:LargeDimS_n} yields
\[
\overline{\chi}_{\overline{S}}(K) \leq \max\set{e^{-2\alpha m}, e^{-\beta n}} \leq e^{-\delta m}\,,
\] 
for some constant $\delta>0$ (we can set $\delta=\min\set{2\alpha, \beta}$). 

Since $\left|\widehat{\GL_k(\FF_q)}\right| \leq \left|\GL_k(\FF_q)\right| \leq q^{k^2}$ and by Lemma \ref{Lemma:SmallDimS_n}, we have
\[
|S| 
\leq \left|\widehat{\GL_k(\FF_q)}\right|^2 \left|\Lambda_c\right|^2 
\leq q^{2k^2} e^{O(\sqrt{n})}\,.
\]
To bound $d_S$, we start with bounding the dimension of each representation in $S$. A representation $\set{\tau\times\lambda, \eta\times \mu}$ in $S$ has dimension 
\[
2d_{\tau\times\lambda}d_{\eta\times \mu} 
=2d_{\tau}d_{\lambda}d_{\eta}d_{\mu} 
\leq 2d_{\tau}d_{\eta} n^{2cn} 
\leq 2q^{k^2}n^{2cn}\,,
\]
where the first inequality follows  Lemma \ref{Lemma:SmallDimS_n}.
The last inequality holds because $d_{\tau}^2 \leq \sum_{\rho\in\widehat{\GL_k(\FF_q)}} d_{\rho}^2=|\GL_k(\FF_q)|$ for any $\tau\in \widehat{\GL_k(\FF_q)}$.
Similarly, a representation $\set{\tau\times \lambda}$ or $\set{\tau\times \lambda}'$ in $S$ has dimension $d_{\tau\times \lambda}^2 \leq q^{k^2}n^{2cn}$. Hence, the maximal dimension of a representation in the set $S$ is
\[
d_S \leq 2q^{k^2}n^{2cn}\,.
\]
Let $4a+4c<d<1$ be a constant and let $\gamma_1$ be any constant such that $0<\gamma_1 < d-4c-4a$. Now we set the dimension threshold $D=n^{dn}$. From the upper bounds on $|S|$ and $d_S$, we get
\begin{align*}
|S|\frac{d_S^2}{D} &\leq 4q^{4k^2}e^{O(\sqrt{n})} n^{(4c-d)n}  \\
&\leq 4e^{O(\sqrt{n})}n^{(4a+4c-d)n} & \mbox{(since $q^{k^2}\leq n^{an}$)}\\
&\leq n^{-\gamma_1 n}								& \mbox{for sufficiently large $n$.}\\
\end{align*}

 Let $L$ be the set of all irreps of $(\GL_k(\mathbb{F}_q)\times S_n)\wr\ints_2$ of dimension at least $D$. Bounding $|L|$ by the number of irreps of $(\GL_k(\mathbb{F}_q)\times S_n)\wr\ints_2$, which is no more than square of the number of irreps of $\GL_k(\mathbb{F}_q)\times S_n$, we have
\[
|L| \leq \left|\widehat{\GL_k(\mathbb{F}_q)}\right|^2  \left|\widehat{S_n}\right|^2 \leq \left|{\GL_k(\mathbb{F}_q)}\right|^2 p(n)^2 \,.
\]
Hence, for sufficiently large $n$,
\[
\begin{split}
\frac{|L|D^2}{\left|(\GL_k(\mathbb{F}_q)\times S_n)\wr\ints_2\right|} 
&\leq \frac{\left|{\GL_k(\mathbb{F}_q)}\right|^2 p(n)^2n^{2dn}}{2\left|(\GL_k(\mathbb{F}_q)\right|^2 |S_n|^2 }
= \frac{ p(n)^2n^{2dn}}{2 (n!)^2 }\\
&\leq \frac{e^{O(\sqrt{n})} n^{2dn}}{2n^{2n}e^{-2n}} \\
&\leq e^{O(n)} n^{2(d-1)n} \leq n^{-\gamma_2 n} \qquad\mbox{so long as $\gamma_2<2(1-d)$.}
\end{split}
\]

By Theorem \ref{Thm:general},  we have
\[
\mathcal{D}_K\leq 4|K|^2(e^{-\delta m}+n^{-\gamma_1 n}+n^{-\gamma_2 n})\leq 4|K|^2(e^{-\delta m}+n^{-\gamma n})\,,
\] 
for some constant $\gamma>0$. This completes the proof.


\end{proof}


\paragraph{Acknowledgments.}  This work was supported by the NSF under grants CCF-0829931, 0835735, and 0829917 and by the DTO under contract W911NF-04-R-0009.  We are grateful to Jon Yard for helpful discussions.


\bibliography{HSP-large-subgroups}

\begin{thebibliography}{23}
\providecommand{\natexlab}[1]{#1}
\providecommand{\url}[1]{\texttt{#1}}
\expandafter\ifx\csname urlstyle\endcsname\relax
  \providecommand{\doi}[1]{doi: #1}\else
  \providecommand{\doi}{doi: \begingroup \urlstyle{rm}\Url}\fi

\bibitem[Bernstein(2008)]{Ref_Bernstein08listdecoding}
Daniel~J. Bernstein.
\newblock List decoding for binary {Goppa} codes, 2008.
\newblock Preprint.

\bibitem[Bernstein et~al.(2008)Bernstein, Lange, and
  Peters]{Ref_Bernstein08attacking}
Daniel~J. Bernstein, Tanja Lange, and Christiane Peters.
\newblock Attacking and defending the {McEliece} cryptosystem.
\newblock In \emph{PQCrypto '08: Proceedings of the 2nd International Workshop
  on Post-Quantum Cryptography}, pages 31--46, Berlin, Heidelberg, 2008.
  Springer-Verlag.
\newblock ISBN 978-3-540-88402-6.

\bibitem[Engelbert et~al.(2007)Engelbert, Overbeck, and
  Schmidt]{Ref_Engelbert07summary}
D.~Engelbert, R.~Overbeck, and A.~Schmidt.
\newblock A summary of {McEliece}-type cryptosystems and their security.
\newblock \emph{J. Math. Crypt.}, 1:\penalty0 151Ð199, 2007.

\bibitem[Fulton and Harris(1991)]{Ref_Fulton91representation}
William Fulton and Joe Harris.
\newblock \emph{Representation Theory - A First Course.}
\newblock Springer-Verlag, New York Inc., 1991.

\bibitem[Grigni et~al.(2004)Grigni, Schulman, Vazirani, and
  Vazirani]{Ref_Grigni04quantum}
Michelangelo Grigni, J.~Schulman, Monica Vazirani, and Umesh Vazirani.
\newblock Quantum mechanical algorithms for the nonabelian hidden subgroup
  problem.
\newblock \emph{Combinatorica}, 24\penalty0 (1):\penalty0 137--154, 2004.

\bibitem[Hallgren et~al.(2006)Hallgren, Moore, R\"{o}tteler, Russell, and
  Sen]{Ref_Hallgren06limitations}
Sean Hallgren, Cristopher Moore, Martin R\"{o}tteler, Alexander Russell, and
  Pranab Sen.
\newblock Limitations of quantum coset states for graph isomorphism.
\newblock In \emph{STOC '06: Proceedings of the thirty-eighth annual ACM
  symposium on Theory of computing}, pages 604--617, 2006.

\bibitem[Kempe and Shalev(2005)]{Ref_Kempe05hidden}
Julia Kempe and Aner Shalev.
\newblock The hidden subgroup problem and permutation group theory.
\newblock In \emph{SODA '05: Proceedings of the sixteenth annual ACM-SIAM
  symposium on Discrete algorithms}, pages 1118--1125, 2005.

\bibitem[Kempe et~al.(2007)Kempe, Pyber, and Shalev]{Ref_Kempe07permutation}
Julia Kempe, Laszlo Pyber, and Aner Shalev.
\newblock Permutation groups, minimal degrees and quantum computing.
\newblock \emph{Groups, Geometry, and Dynamics}, 1\penalty0 (4):\penalty0
  553--584, 2007.
\newblock URL \url{http://xxx.lanl.gov/abs/quant-ph/0607204}.

\bibitem[Loidreau and Sendrier(2001)]{Ref_Loidreau01weak}
Pierre Loidreau and Nicolas Sendrier.
\newblock Weak keys in the {McEliece} public-key cryptosystem.
\newblock \emph{IEEE Transactions on Information Theory}, 47\penalty0
  (3):\penalty0 1207--1212, 2001.

\bibitem[Lomont(2004)]{Ref_Lomont04hidden}
Chris Lomont.
\newblock The hidden subgroup problem - review and open problems, 2004.
\newblock URL \url{arXiv.org:quant-ph/0411037}.

\bibitem[McEliece(1978)]{Ref_McEliece78public}
R.J. McEliece.
\newblock A public-key cryptosystem based on algebraic coding theory.
\newblock \emph{JPL DSN Progress Report}, pages 114--116, 1978.

\bibitem[Menezes et~al.(1996)Menezes, van Oorschot, and
  Vanstone]{Ref_Menezes96handbook}
A.J. Menezes, P.C. van Oorschot, and S.A. Vanstone.
\newblock \emph{Handbook of applied cryptography}.
\newblock CRC Press, 1996.

\bibitem[Moore et~al.(2008)Moore, Russell, and Schulman]{Ref_Moore08symmetric}
Cristopher Moore, Alexander Russell, and Leonard~J. Schulman.
\newblock The symmetric group defies strong quantum {Fourier} sampling.
\newblock \emph{SIAM Journal of Computing}, 37:\penalty0 1842--1864, 2008.

\bibitem[Petrank and Roth(1997)]{Ref_Petrank97code}
E.~Petrank and R.M. Roth.
\newblock Is code equivalence easy to decide?
\newblock \emph{IEEE Transactions on Information Theory}, 43\penalty0
  (5):\penalty0 1602 -- 1604, 1997.
\newblock \doi{10.1109/18.623157}.

\bibitem[Regev(2005)]{regev-jacm}
Oded Regev.
\newblock On lattices, learning with errors, random linear codes, and
  cryptography.
\newblock In \emph{STOC '05: Proceedings of the thirty-seventh annual ACM
  symposium on Theory of computing}, pages 84--93, 2005.

\bibitem[Roichman(1996)]{Ref_Roichman96upper}
Yuval Roichman.
\newblock Upper bound on the characters of the symmetric groups.
\newblock \emph{Invent. Math.}, 125\penalty0 (3):\penalty0 451--485, 1996.

\bibitem[Ryan(2007)]{Ref_Ryan07excluding}
John~A. Ryan.
\newblock Excluding some weak keys in the {McEliece} cryptosystem.
\newblock In \emph{Proceedings of the 8th IEEE Africon}, pages 1--5, 2007.

\bibitem[Shor(1997)]{Ref_Shor97polynomial}
Peter.~W. Shor.
\newblock Polynomial-time algorithms for prime factorization and discrete
  logarithms on a quantum computer.
\newblock \emph{SIAM Journal on Computing}, 26:\penalty0 1484--1509, 1997.

\bibitem[Sidelnikov and Shestakov(1992)]{Ref_Sidelnikov92insecurity}
V.~M. Sidelnikov and S.~O. Shestakov.
\newblock On insecurity of cryptosystems based on generalized {R}eed-{S}olomon
  codes.
\newblock \emph{Discrete Mathematics and Applications}, 2\penalty0
  (4):\penalty0 439--444, 1992.

\bibitem[Simon(1997)]{Ref_Simon97power}
Daniel~R. Simon.
\newblock On the power of quantum computation.
\newblock \emph{SIAM J. Comput.}, 26\penalty0 (5):\penalty0 1474--1483, 1997.

\bibitem[Stichtenoth(1990)]{Ref_Stichtenoth90on}
Henning Stichtenoth.
\newblock On automorphisms of geometric {Goppa} codes.
\newblock \emph{Journal of Algebra}, 130:\penalty0 113--121, 1990.

\bibitem[Stichtenoth(2008)]{Ref_Stichtenoth08algebraic}
Henning Stichtenoth.
\newblock \emph{Algebraic Function Fields and Codes}.
\newblock Springer, 2nd edition, 2008.

\bibitem[van Lint(1992)]{Ref_vanLint92introduction}
J.H van Lint.
\newblock \emph{Introduction to coding theory}.
\newblock Springer-Verlag, 2nd edition, 1992.

\end{thebibliography}

\newpage
\appendix \textbf{\Large{Appendices} }

\section{Supplemental proofs for the main theorem}\label{ApxSec:MainThm}

\begin{proposition}\label{Prop:ExpSchur} Let $H<G$ and $g$ be chosen from $G$ uniformly at random. Then for $\rho\in\widehat{G}$ and $\mathbf{b}\in B_{\rho}$, 
\[
\expect[g]{P_{H^g}(\mathbf{b}\mid\rho)}={1}/{d_{\rho}}\,.
\]
\end{proposition}
\begin{proof}
Schur's lemma asserts that if $\rho$ is irreducible, the only matrices which commute with $\rho(g)$ for all $g$ are the scalars. Hence, $$\expect[g]{\Pi^{\rho}_{H^g}}=\frac{1}{|G|}\sum_{g\in G}\rho^{\dagger}(g)\Pi^{\rho}_{H}\rho(g)=\frac{\tr(\Pi^{\rho}_{H})}{d_{\rho}}\mathbf{1}_{d_{\rho}}\,,$$ which implies that
\[
\expect[g]{\|\Pi^{\rho}_{H^g}\mathbf{b}\|^2} 
= \expect[g]{\tup{\mathbf{b}, \Pi^{\rho}_{H^g}\mathbf{b}}}
=\tup{\mathbf{b}, \expect[g]{\Pi^{\rho}_{H^g}}\mathbf{b}}
= \frac{\tr(\Pi^{\rho}_{H})}{d_{\rho}}\,.
\]
\end{proof}


\subsection{Proof of Lemma~\ref{Lemma:LargeSmall}}

\begin{proof}[Proof of Lemma \ref{Lemma:LargeSmall}]
Let $L_{\sigma}$ be the subspace of $\rho\otimes \rho^*$ consisting of all copies of $\sigma$. Since 
$B_{\rho}$ is orthonormal, the vectors $\set{\mathbf{b}\otimes \mathbf{b}^*\mid \mathbf{b}\in B_{\rho}}$ are mutually orthogonal in $\rho\otimes \rho^*$.
Thus, 
\[
\sum_{\mathbf{b}\in B_{\rho}}\left\|\Pi_{\sigma}^{\rho\otimes\rho^*}(\mathbf{b}\otimes \mathbf{b}^*)\right\|^2\leq \dim L_{\sigma}\,.
\]
Note that $\dim L_{\sigma}$ is equal to $d_{\sigma}$ times \text{ the multiplicity of $\sigma$ in $\rho\otimes \rho^*$}. 
On the other hand, we have
\[
\begin{split}
\text{multiplicity of $\sigma$ in $\rho\otimes \rho^*$} 
&=\tup{\chi_{\sigma}, \chi_{\rho}\chi_{\rho^*}}=\tup{\chi_{\sigma}\chi_{\rho}, \chi_{\rho^*}}\\
&=\text{multiplicity of $\rho^*$ in $\sigma\otimes \rho$}\\
&\leq \frac{\dim (\sigma\otimes \rho)}{\dim\rho^*} = d_{\sigma}\,,
\end{split}
\]
Hence, 
\[
\sum_{\mathbf{b}\in B_{\rho}}\left\|\Pi_{\sigma}^{\rho\otimes\rho^*}(\mathbf{b}\otimes \mathbf{b}^*)\right\|^2\leq d^2_{\sigma}\,.
\]

\end{proof}

\section{Rational Goppa codes}\label{ApxSec:Goppa}

This part summarizes definitions and key properties of rational Goppa codes that would be useful in our analysis. Following Stichtenoth \citep{Ref_Stichtenoth90on}, we shall describe Goppa codes in terms of algebraic function fields instead of algebraic curves. A complete treatment for this subject can be found in \citep{Ref_Stichtenoth08algebraic}.

A \emph{rational function field} over $\FF_q$ is a field extension $\FF_q(x)/\FF_q$ for some $x$ transcendental over $\FF_q$. Each element $z\in \FF_q(x)$ can be viewed as a function whose evaluation at a base field element $a\in \FF_q$ is determined as follows: write $z=f(x)/g(x)$ for some polynomials $f(x), g(x)\in \FF_q[x]$, then
\[
z(a) =
\begin{cases}
\frac{f(a)}{g(a)} \in \FF_q &\text{if}~ g(a)\neq 0 \\
\infty &\text{if}~ g(a)= 0\,. \\
\end{cases}
\]

 A \emph{Rieman-Roch space}\footnote{In terms of algebraic function fields, a Rieman-Roch space is defined in the association with a \emph{divisor} of the function field $F/K$, where a divisor is a finite sum $\sum_i n_i P_i$ with $n_i\in\ints$ and $P_i$'s being \emph{places} of the function field. In the rational function field $K(x)/K$, we can show that every divisor can be written as $rP_{\infty}+(z)$ for some integer $r$ and some nonzero $z\in K(x)$, where $P_{\infty}$ is the infinite place (defined in \citep[pg. 9]{Ref_Stichtenoth08algebraic}), and $(z)$ is the \emph{principal divisor} of $z$. The space $\mathcal{L}(r, g, h)$ is indeed the Rieman-Roch space associated with the divisor $rP_{\infty}+(z)$ with $z=h(x)/g(x)$.} 
 in the rational function field $\FF_q(x)/\FF_q$ is a subset of $\FF_q$ of the form
\[
\mathcal{L}(r, g, h) = \set{\frac{f(x)g(x)}{h(x)}\Bigl| f(x)\in \FF_q[x],~ \deg f(x)\leq r}
\]
for some nonzero polynomials $g(x), h(x)\in K[x]$ and some integer $r$. Note that $\mathcal{L}(r, g, h)$ is a vector space of dimension $r+1$ over $\FF_q$.

\begin{definition}(A special case of Definition 2.2.1 in \citep{Ref_Stichtenoth08algebraic})\label{Def:rational-Goppa-code}
Let $g(x), h(x)\in\FF_q[x]$ be nonzero coprime polynomials, and let $r<n$ a nonnegative integer. 
Let $\gamma_1,\ldots,\gamma_n$ be $n$ distinct elements in the field\footnote{In the case $r=\deg h(x)-\deg g(x)$, one can choose one of the points $P_i$'s to be $\infty$. However, we rule out this case  to keep the discussion simple.} $\FF_q$ such that $g(\gamma_i)\neq 0$ and $h(\gamma_i)\neq 0$ for all $i$. Then a \emph{rational Goppa code} associated with $g, h$ and $\gamma_i$'s is defined by 
\[
\mathcal{C}(\gamma_1,\ldots,\gamma_n, r, g, h) \eqdef \set{(z(\gamma_1),\ldots, z(\gamma_n))\mid z\in\mathcal{L}(r,g,h)}\subset \FF_q^n\,.
\]
\end{definition}

\begin{remark}
A classical binary Goppa code can be obtained by setting $q=2^m$, $r=n-\deg g(x)-1$, and $h(x)=\sum_{j=1}^n\prod_{i\neq j}(x-\gamma_i)$  and then intersecting the code $\mathcal{C}(\gamma_1,\ldots,\gamma_n, r, g, h)$ with the vector space $\FF_2^n$ (see \citep{Ref_Bernstein08listdecoding}). Generalized Reed-Solomon codes are a special case of rational Goppa codes in which the polynomials $g(x)$ and $h(x)$ are both constants.
\end{remark}

\begin{theorem}\label{Thm:rational-Goppa-code}(A special case of Corollary 2.2.3 in \citep{Ref_Stichtenoth08algebraic})
The code defined in Definition \ref{Def:rational-Goppa-code} is an $[n,k,d]$-linear code over $\FF_q$ with dimension $k=r+1$ and minimum distance $d\geq n-r$. Consequentially, this code can correct  at least $(n-r-1)/2$ errors.
\end{theorem}

The rational Goppa code $\mathcal{C}(\gamma_1,\ldots,\gamma_n, r, g, h)$  has a generator matrix:
\[
M_0= \begin{pmatrix}
\frac{g(\gamma_1)}{h(\gamma_1)} & \ldots & \frac{g(\gamma_n)}{h(\gamma_n)}\\
\gamma_1^{}\frac{g(\gamma_1)}{h(\gamma_1)} & \ldots & \gamma_n^{}\frac{g(\gamma_n)}{h(\gamma_n)}\\
\vdots   & \ddots &\vdots\\
\gamma_1^{r}\frac{g(\gamma_1)}{h(\gamma_1)} & \ldots & \gamma_n^{r}\frac{g(\gamma_n)}{h(\gamma_n)}\\
\end{pmatrix}\,.
\]
\begin{proposition}
The matrix $M_0$ has full rank, that is, its column rank equals $r+1$. Hence, every generator matrix of a rational Goppa code has full rank.
\end{proposition}
\begin{proof}
It suffices to to show that the first $r+1$ columns of $M_0$ are linearly independent. Equivalently, we show that the matrix $N_0$ below has nonzero determinant:
\[
N_0=
\begin{pmatrix}
\frac{g(\gamma_1)}{h(\gamma_1)} & \ldots & \frac{g(\gamma_{r+1})}{h(\gamma_{r+1})}\\
\gamma_1^{}\frac{g(\gamma_1)}{h(\gamma_1)} & \ldots & \gamma_{r+1}^{}\frac{g(\gamma_{r+1})}{h(\gamma_{r+1})}\\
\vdots   & \ddots &\vdots\\
\gamma_1^{r}\frac{g(\gamma_1)}{h(\gamma_1)} & \ldots & \gamma_{r+1}^{r}\frac{g(\gamma_{r+1})}{h(\gamma_{r+1})}\\
\end{pmatrix} 
=
\begin{pmatrix}
1 & \ldots & 1\\
\gamma_1^{} & \ldots & \gamma_{r+1}^{}\\
\vdots   & \ddots &\vdots\\
\gamma_1^{r} & \ldots & \gamma_{r+1}^{r}\\
\end{pmatrix}
\begin{pmatrix}
\frac{g(\gamma_1)}{h(\gamma_1)} & ~ & ~\\
~& \ddots & ~\\
~& ~  & \frac{g(\gamma_{r+1})}{h(\gamma_{r+1})}\\
\end{pmatrix}\,.
\]
The first matrix in the above product is a Vandermonde matrix, which has nonzero determinant because $\gamma_i$'s are distinct. The second matrix also has nonzero determinant because $g(\gamma_i)\neq 0$ for all $i$.
Hence, $N_0$ has nonzero determinant.
\end{proof}

An important property of rational Goppa codes is that in general their automorphisms are induced by projective transformations of the projective line. We will make this precise below.

\begin{definition}(See \citep[pg. 53]{Ref_vanLint92introduction}) Let $C$ be a code of length $n$. An \emph{automorphism} of $C$ is a permutation $\pi\in S_n$ which maps every word in $C$ to a word in $C$ by acting on the positions of the codewords. The set of all automorphisms of $C$ forms a group called the \textbf{automorphism group} of $C$.
\end{definition}

In particular, an automorphism of $\mathcal{C}(\gamma_1,\ldots,\gamma_n, r, g, h)$ is a permutation $\pi\in S_n$ such that
\[
\mathcal{C}(\gamma_1,\ldots,\gamma_n, r, g, h) = \mathcal{C}(\gamma_{\pi(1)},\ldots,\gamma_{\pi(n)}, r, g, h)\,.
\] 
\begin{remark}
Suppose $M$ is a generator matrix for an $[n,k]$-linear code $C$ over $\FF_q$. Then a permutation $\pi\in S_n$ is an automorphism of $C$ if and only if there is an invertible matrix $A\in \GL_k(\FF_q)$ such that $AMP_{\pi}=P_{\pi}$, where $P_{\pi}$ denotes the permutation  matrix corresponding to $\pi$. If $M$ has full rank, there is exactly one such matrix $A$ for each automorphism $\pi$ of $C$.
\end{remark}
\begin{theorem}[Stichtenoth \citep{Ref_Stichtenoth90on}] Suppose $1\leq r \leq n-3$. Then the automorphism group of the rational Goppa code $\mathcal{C}(\gamma_1,\ldots,\gamma_n, r, g, h)$ is isomorphic to a subgroup of $\aut(\FF_q(x)/\FF_q)$.
\end{theorem}
\begin{fact}
The automorphism group of the rational function field $\FF_q(x)/\FF_q$ is isomorphic to the projective linear group over $\FF_q$. In notations, $\aut(\FF_q(x)/\FF_q)\simeq \PGL_2(\FF_q)$.
\end{fact}

Let $C=\mathcal{C}(\gamma_1,\ldots,\gamma_n, r, g, h)$ be a rational Goppa code. To give an intuition for how  the automorphism group of $C$ is embedded in  $\PGL_2(\FF_q)$, consider a transformation $\sigma \in \PGL_2(\FF_q)$ and view each element $a\in\FF_q$ as the projective line $[a:1]$ (the point at infinity is written as $[1:0]$). Suppose $\sigma$ transforms $[a:1]$ to the projective line $[b:1]$, then we shall write $\sigma a=b$. If $\sigma$ transforms each line $[\gamma_i:1]$ to some line $[\gamma_j:1]$, then $\sigma$ induces another rational Goppa code:
\[
\mathcal{C}(\sigma\gamma_1,\ldots,\sigma\gamma_n, r, g, h)\,.
\]
If, further, $\mathcal{C}(\sigma\gamma_1,\ldots,\sigma\gamma_n, r, g, h)$ equals the original code $C$, then $\sigma$ induces an automorphism of $C$. Stichtenoth's theorem establishes that every automorphism of $C$ is induced by such a transformation in $\PGL_2(\FF_q)$.


\section{Supplemental proofs for $\GL_2(\mathbb{F}_q)$}\label{ApxSec:GL_2}

\subsection{Irreducible representations of $\GL_{2}(\mathbb{F}_q)$}
In this part, we will first present a preliminary background on the structure of $\GL_{2}(\mathbb{F}_q)$ followed by description of its irreps. We refer readers to \citep[\S 5.2]{Ref_Fulton91representation} for the missing technical details in this part.  

Viewing $\GL_{2}(\mathbb{F}_q)$ as the group of all $\mathbb{F}_q$-linear invertible endomorphisms of the quadratic extension $\mathbb{F}_{q^2}$ of $\mathbb{F}_q$, we have a large subgroup of $\GL_{2}(\mathbb{F}_q)$ that is isomorphic to  $\mathbb{F}^*_{q^2}$ via the identification:
\[
\set{f_{\xi} \mid \xi\in\mathbb{F}^*_{q^2}} \simeq \mathbb{F}^*_{q^2}\,, \quad f_{\xi} \leftrightarrow \xi 
\]
where $f_{\xi}:\mathbb{F}_{q^2}\to \mathbb{F}_{q^2}$ is the $\mathbb{F}_q$-linear map given by $f_{\xi}(\nu)=\xi \nu$ for all $\nu\in\mathbb{F}_{q^2}$.

To turn each map $f_{\xi}$ into a matrix form, we fix a basis $\set{1,\gamma}$ of $\mathbb{F}_{q^2}$ as a vector space over $\mathbb{F}_{q}$. For each $\xi\in \mathbb{F}_{q^2}$, writing $\xi=\xi_{x,y}=x+\gamma y$ for some $x,y\in\mathbb{F}_q$, then the map $f_{\xi}$ corresponds to the matrix 
$\begin{pmatrix}
 x     &   \gamma^2 y \\
 y     &   x
\end{pmatrix}$, since $f_{\xi}(1)=x+\gamma y$ and $f_{\xi}(\gamma)=\gamma^2y +\gamma x$. 
Hence, we can rewrite the above identification as
\[
\set{\begin{pmatrix}
 x     &   \gamma^2 y \\
 y     &   x
\end{pmatrix}\mid x,y \in \mathbb{F}_{q}, x\neq 0 ~\text{or}~ y\neq 0} 
\simeq \mathbb{F}^*_{q^2}\,, 
\quad
\begin{pmatrix}
 x     &   \gamma^2 y \\
 y     &   x
\end{pmatrix} \leftrightarrow \xi_{x,y} = x+\gamma y\,. 
\] 

For example, if $q$ is odd, choose a generator $\epsilon$ of $\mathbb{F}^*_q$, then $\epsilon$ must be non-square in $\mathbb{F}_q$, which implies that $\set{1,\sqrt{\epsilon}}$ form a basis of $\mathbb{F}_{q^2}$ as a vector space over $\mathbb{F}_q$. In such a case, we can define $\xi_{x,y}=x+y\sqrt{\epsilon}$.

\paragraph{Conjugacy classes.} Group $\GL_{2}(\mathbb{F}_q)$ has four types of conjugacy classes in Table \ref{Tb:ConjugacyGL_2q}, with representatives described as follows:
\[
a_x= \begin{pmatrix}
 x     &   0 \\
 0     &   x
\end{pmatrix} 
\quad
b_x= \begin{pmatrix}
 x     &   1 \\
 0     &   x
\end{pmatrix} 
\quad
c_{x,y}= \begin{pmatrix}
 x     &   0 \\
 0     &   y
\end{pmatrix} 
\quad
d_{x,y}= \begin{pmatrix}
 x     &   \gamma^2 y \\
 y     &   x
\end{pmatrix} 
\]

\begin{table}[h]
  \centering 
\begin{tabular}{l|l|l|l|l}
class &  $[a_x]$ & $[b_x]$ & $[c_{x,y}]=[c_{y,x}]$ & $[d_{x,y}]=[d_{x,-y}]$ \\
~     &  $x\in\mathbb{F}_q^*$ &  $x\in\mathbb{F}_q^*$ & $x,y\in\mathbb{F}_q^*, x\neq y$  &   $x\in\mathbb{F}_q, y\in\mathbb{F}_q^*$ \\
\hline
class size       &  1 & $q^2-1$  & $q^2+q$ & $q^2-q$ \\
\hline
no. of classes   &  $q-1$ & $q-1$ & $\frac{(q-1)(q-2)}{2}$ & $\frac{q(q-1)}{2}$\\
\end{tabular}
\begin{quote}  
  \caption{Conjugacy classes of $\GL_{2}(\mathbb{F}_q)$, where $[g]$ denotes the class of representative $g$. }
\end{quote}
  \label{Tb:ConjugacyGL_2q}
\end{table}
There are $q^2-1$ conjugacy classes, hence there are exactly $q^2-1$ irreps of $\GL_{2}(\mathbb{F}_q)$.
We shall briefly describe below how to construct all those representations.

\paragraph{Linear representations.}

For each character $\alpha:\mathbb{F}_q^*\to\coms^*$ of the cyclic group $\mathbb{F}_q^*$, we have a one-dimensional representation $U_{\alpha}$ of  $\GL_{2}(\mathbb{F}_q)$ defined by:
\[
U_{\alpha}(g) = \alpha(\det(g)) \quad\quad \forall g\in\GL(2,q)\,.
\]

To compute $U_{\alpha}(d_{x,y})$, we shall use the following fact:
\[
\det \begin{pmatrix}
  x    &  \gamma^2 y   \\
  y    &  x
\end{pmatrix} =\mathrm{Norm}_{\mathbb{F}_{q^2}/\mathbb{F}_q}(\xi_{x,y}) = \xi_{x,y}\cdot \xi_{x,y}^q = \xi_{x,y}^{q+1}\,.
\]
Recall that there are $q-1$ characters of $\mathbb{F}_q^*=\tup{\epsilon}$ corresponding to $q-1$ places where the generator $\epsilon$ can be sent to. The linear representation $U_{\alpha_0}$, where $\alpha_0$ is the character sending $\epsilon$ to 1, is indeed the trivial representation, denoted $U$. 

\paragraph{Irreducible representations by action on $\mathbb{P}^1(\mathbb{F}_q)$.}
$\GL_2(\mathbb{F}_q)$ acts transitively on the projective line $\mathbb{P}^1(\mathbb{F}_q)$ in the natural way:
\[
\begin{pmatrix}
  a    &  b   \\
  c    &  d
\end{pmatrix} \cdot [x:y] = 
\begin{pmatrix}
  a    &  b   \\
  c    &  d
\end{pmatrix}\begin{bmatrix}
  x     \\
  y    
\end{bmatrix} = [ax+by:cx+dy]\,,
\]
in which the stabilizer of the infinite point $[1:0]$ is the Borel subgroup $B$:
\[
B=\set{
\begin{pmatrix}
  a    &  b   \\
  0    &  d
\end{pmatrix} \mid a,d\in \mathbb{F}_q^*,~ b\in \mathbb{F}_q }\,.
\]

The permutation representation of $\GL_{2}(\mathbb{F}_q)$ given by this action on $\mathbb{P}^1(\mathbb{F}_q)$ has dimension $q+1$ and  decomposes into the trivial representation $U$ and a $q$-dimensional representation $V$. The character of $V$ is given as follows:
\[
\chi_V(a_x) = q \quad \chi_V(b_x)=0 \quad \chi_V(c_{x,y})=1 \quad \chi_V(d_{x,y})=-1\,.
\]
By checking $\tup{\chi_V,\chi_V}=1$, we see that $V$ is irreducible. 
Hence, for each of the $q-1$ characters $\alpha$ of $\mathbb{F}^*_q$, we have a $q$-dimensional irrep $V_{\alpha}=V\otimes U_{\alpha}$. Note that $V=V\otimes U$.

\paragraph{Irreducible representations induced from Borel subgroup $B$.} For each pair of characters $\alpha, \beta$ of $\mathbb{F}_q^*$, there is a character of the subgroup $B$:
\[
\phi_{\alpha,\beta} : B\to \coms^* 
\quad\text{by }
\begin{pmatrix}
  a    &  b   \\
  0    &  d
\end{pmatrix} \mapsto \alpha(a)\beta(d)\,.
\]
In other words, $\phi_{\alpha,\beta}$ is a one-dimensional representation of subgroup $B$. Let $W_{\alpha,\beta}$ be the representation of $\GL_{2}(\mathbb{F}_q)$ induced by $\phi_{\alpha,\beta}$. By computing characters, we have 
\begin{itemize}
\item $W_{\alpha,\beta}=W_{\beta,\alpha}$, 
\item $W_{\alpha,\alpha}=U_{\alpha}\oplus V_{\alpha}$, and 
\item $W_{\alpha,\beta}$ is irreducible for $\alpha\neq \beta$. Each of these representations has dimension equal the index of $B$ in $\GL_{2}(\mathbb{F}_q)$, i.e., $[\GL(2,q): B]=q+1$.
\end{itemize}

There are $((q-1)^2-(q-1))/2=(q-1)(q-2)/2$ distinct irreps of this type.

\paragraph{Irreducible representations by characters of $\mathbb{F}_{q^2}^*$.}

Let $\varphi: \mathbb{F}_{q^2}^*\to \coms^*$ be a character of the cyclic group $\mathbb{F}_{q^2}^*$. Since $\mathbb{F}_{q^2}^*$ can be viewed as a subgroup of $\GL_{2}(\mathbb{F}_q)$, we have the induced representation $\Ind \varphi$, which  is not irreducible. However, it gives us a $(q-1)$-dimensional irrep with character given by
\[
\chi_{\varphi} = \chi_{V\otimes W_{\alpha,1}}-\chi_{W_{\alpha,1}}-\chi_{\Ind\varphi}\quad\text{if $\varphi|_{\mathbb{F}^*_q}=\alpha$.}
\]
Note that $\Ind\varphi \simeq \Ind\varphi^q$, thus $X_{\varphi}\simeq X_{\varphi^q}$. So, the characters $\varphi$ of $\mathbb{F}_{q^2}^*$ with $\varphi\neq \varphi^q$ give a rise to the $\frac{1}{2}q(q-1)$ remaining irreps of $\GL_{2}(\mathbb{F}_q)$.

A summary of all irreducible characters of $\GL_{2}(\mathbb{F}_q)$ is given in Table \ref{Table:CharacterGL_2_q} below.

\begin{table}[h]
  \centering 
\begin{tabular}{l|l|l|l|l|l}
$\rho$       & $d_{\rho}$&  $\chi_{\rho}(a_x)$ & $\chi_{\rho}(b_x)$ & $\chi_{\rho}(c_{x,y})$ & $\chi_{\rho}(d_{x,y})$ \\
\hline
$U_{\alpha}$ & 1 & $\alpha(x^2)$ & $\alpha(x^2)$ & $\alpha(xy)$ & $\alpha(\xi_{x,y}^{q+1})$\\
\hline
$V_{\alpha}$ & $q$ & $q\alpha(x^2)$ & 0 & $\alpha(xy)$ & $-\alpha(\xi_{x,y}^{q+1})$\\
\hline
$W_{\alpha,\beta}$ ($\alpha\neq\beta$)& $q+1$& $(q+1)\alpha(x)\beta(x)$ & $\alpha(x)\beta(x)$ & $\alpha(x)\beta(y)+\alpha(y)\beta(x)$ & 0 \\
\hline
$X_{\varphi}$ & $q-1$ & $(q-1)\varphi(x)$ & $-\varphi(x)$ & 0 & $-(\varphi(\xi_{x,y})+\varphi(\xi_{x,y}^q))$\\
\end{tabular}
\begin{quote}  
  \caption{Character table of $\GL_{2}(\mathbb{F}_q)$, where $\alpha, \beta$ are characters of $\mathbb{F}_q^*$, and $\varphi$ is a character of $\mathbb{F}_{q^2}^*$ with $\varphi^q\neq \varphi$, and $d_{\rho}=\chi_{\rho}(a_1)$ is the dimension of $\rho$.}
\end{quote}  
  \label{Table:CharacterGL_2_q}
\end{table}

\subsection{Proof of Lemma~\ref{Lemma:LinearRep_0}}

In the remaining of this section, we devote to prove Lemma \ref{Lemma:LinearRep_0}, which states that there are at most two linear representations appearing in the decomposition of $\rho\otimes \rho^*$, for any irrep $\rho$ of $\GL_{2}(\mathbb{F}_q)$. Obviously, if $\rho$ is linear then $\rho\otimes \rho^*$ is the trivial representation. Therefore, we shall only consider the cases where $\rho$ is non-linear. 

Recall that the multiplicity of $U_{\alpha}$ in $\rho\otimes \rho^*$ is given by 
\[
\tup{\chi_{\rho\otimes \rho^*},\chi_{U_{\alpha}}} = \frac{1}{|G|}\sum_{g\in G} |\chi_{\rho}(g)|^2\chi_{U_{\alpha}}(g)= \frac{1}{|G|}(A(\rho,\alpha)+B(\rho,\alpha)+C(\rho,\alpha)+D(\rho,\alpha))\,,
\]
where $A(\rho,\alpha), B(\rho,\alpha), C(\rho,\alpha), D(\rho,\alpha))$ are the sum of $|\chi_{\rho}(g)|^2\chi_{U_{\alpha}}(g)$ over all element $g$ in the conjugacy classes with representatives of the form $a_x, b_x, c_{x,y}$ and $d_{x,y}$, respectively. That is, from the description of conjugacy classes in Table~\ref{Tb:ConjugacyGL_2q},
\begin{align*}\label{}
A(\rho,\alpha)&=\sum_{x\in\mathbb{F}_q^*} |\chi_{\rho}(a_x)|^2\chi_{U_{\alpha}}(a_x)\\
B(\rho,\alpha)&=(q^2-1)\sum_{x\in\mathbb{F}_q^*} |\chi_{\rho}(b_x)|^2\chi_{U_{\alpha}}(b_x)\\
C(\rho,\alpha)&=\frac{1}{2}(q^2+q)\sum_{x,y\in\mathbb{F}_q^*, x\neq y} |\chi_{\rho}(c_{x,y})|^2\chi_{U_{\alpha}}(c_{x,y})\\
D(\rho,\alpha)&=\frac{1}{2}(q^2-q)\sum_{x,y\in\mathbb{F}_q, y\neq 0} |\chi_{\rho}(d_{x,y})|^2\chi_{U_{\alpha}}(d_{x,y})\,.
\end{align*}

\remove{
Then, from the description of conjugacy classes of $\GL_{2}(\mathbb{F}_q)$ in Table \ref{}, we have
\[
\tup{\chi_{\rho\otimes \rho^*},\chi_{U_{\alpha}}} =\frac{1}{|G|}\left(A(\rho,\alpha)+(q^2-1)B(\rho,\alpha)+\frac{1}{2}(q^2+q)C(\rho,\alpha)+\frac{1}{2}(q^2-q)D(\rho,\alpha)\right) 
\]
}
Our goal below will be to show that $\tup{\chi_{\rho\otimes \rho^*},\chi_{U_{\alpha}}}=0$ for all but two linear representations $U_{\alpha}$ and for all non-linear irrep $\rho$ of $\GL_{2}(\mathbb{F}_q)$. We begin with the following lemma.
  
\begin{lemma}\label{Lemma:Character_0}
Let $F$ be a finite field and $\phi:F^{\times}\to\coms^*$ be a non-trivial character of the cyclic group $F^{\times}$, i.e., $\phi(x)\neq 1$ for some $x$. 
Then $\sum_{x\in F^{\times}} \phi(x)=0$.
\end{lemma}
\begin{proof}
Let $n$ be the order of $F^{\times}$ and let $\tau$ be a generator of $F^{\times}$. Then $\tau^n=1$ which implies $\phi(\tau)^n=1$. Since $\phi$ is non-trivial, we must have $\phi(\tau)\neq 1$. Hence,
\[
\sum_{x\in F^{\times}} \phi(x) = \sum_{k=0}^{n-1} \phi(\tau^k) =\sum_{k=0}^{n-1} \phi(\tau)^k = \frac{\phi(\tau)^n-1}{\phi(\tau)-1} =0\,.
\]
\end{proof}

Note that for any character $\alpha$ of $\mathbb{F}^*_q$, the map $\alpha^2:\mathbb{F}^*_q\to\coms^*$ defined by $\alpha^2(x)= \alpha(x^2)$ is also a character of $\mathbb{F}^*_q$. 
Hence, we have the following direct corollaries of Lemma \ref{Lemma:Character_0}.

\begin{corollary}\label{Cor:Character_0}
Let $\alpha$ be a character of $\mathbb{F}^*_q$ such that $\alpha^2$ is non-trivial. Then $\sum_{x\in\mathbb{F}^*_q}\alpha(x^2)=0$.
\end{corollary}

\begin{corollary}\label{Cor:A_B_0}
Let $\rho$ be an irrep of $\GL_{2}(\mathbb{F}_q)$ and let $\alpha$ be a character of $\mathbb{F}^*_q$ such that $\alpha^2$ is non-trivial. Then we always have $A(\rho,\alpha)=B(\rho,\alpha)=0$.
\end{corollary}
\begin{proof}
Observe that $|\chi_{\rho}(a_x)|$ and $|\chi_{\rho}(b_x)|$ do not depend on $x$, and $\chi_{U_{\alpha}}(a_x)=\chi_{U_{\alpha}}(b_x)=\alpha(x^2)$. Hence, to show  $A(\rho,\alpha)=B(\rho,\alpha)=0$, it suffices to use the fact that $\sum_{x\in\mathbb{F}^*_q}\alpha(x^2)=0$. 
\end{proof}

\begin{remark}
There are at most two characters $\alpha$ of $\mathbb{F}^*_q$ such that $\alpha^2$ is trivial. They are the trivial one, and the one that maps $\epsilon \to \omega^{\frac{q-1}{2}}$ if $q$ is odd, where $\omega=e^{\frac{2\pi i}{q-1}}$ is a primitive $(q-1)^{\rm th}$ root of unity, and $\epsilon$ is a chosen generator of the cyclic group $\mathbb{F}^*_q$. To see this, suppose  $\alpha(\epsilon)=\omega^k$, for some $k\in\set{0,1,\ldots, q-2}$. If $\alpha(\epsilon)^2=1$, then $\omega^{2k}=1$, which implies $q-1\mid 2k$ because $\omega$ has order $q-1$. Hence $2k\in \set{0, q-1}$.
\end{remark}

With this remark, Lemma \ref{Lemma:LinearRep_0} will immediately follows Lemma \ref{Lemma:GL_2_q_LinearRep_0} below.

\begin{lemma}\label{Lemma:GL_2_q_LinearRep_0} Let $\rho$ be a non-linear irrep of $\GL_{2}(\mathbb{F}_q)$ and let $\alpha$ be a character of $\mathbb{F}^*_q$ such that  $\alpha^2$ is trivial. Then $U_{\alpha}$ does not appear in the decomposition of $\rho\otimes \rho^*$.
\end{lemma}
\begin{proof} We will prove case by case of $\rho$ that $C(\rho,\alpha)=D(\rho,\alpha)=0$, which, together with Corollary \ref{Cor:A_B_0}, will complete the proof for the lemma.
 
\paragraph{Case $\rho=W_{\beta,\beta'}$.}
For this case, as $|\chi_{W_{\beta,\beta'}}(d_{x,y})|=0$, we only need to show $C(W_{\beta,\beta'},\alpha)=0$.
Considering $x,y\in\mathbb{F}^*_q$ with $x\neq y$ and letting $z=x^{-1}y\neq 1$, we have
\[
\begin{split}
|\chi_{W_{\beta,\beta'}}(c_{x,y})|^2
&= [\beta(x)\beta'(y)+\beta(y)\beta'(x)][\beta(x^{-1})\beta'(y^{-1})+\beta(y^{-1})\beta'(x^{-1})]\\
&=2+\beta(xy^{-1})\beta'(yx^{-1})+\beta(yx^{-1})\beta'(xy^{-1})\\
&=2+\beta(z^{-1})\beta'(z)+\beta(z)\beta'(z^{-1})\\
\end{split}
\]
This means $|\chi_{W_{\beta,\beta'}}(c_{x,y})|^2$ only depends on $z=x^{-1}y$.  
Now let $\gamma(z)=|\chi_{W_{\beta,\beta'}}(c_{x,y})|^2\alpha(z)$, we have
\[
\begin{split}
|\chi_{W_{\beta,\beta'}}(c_{x,y})|^2\chi_{U_{\alpha}}(c_{x,y}) =|\chi_{W_{\beta,\beta'}}(c_{x,y})|^2\alpha(x^2z)=\gamma(z)\alpha(x^2).
\end{split}
\]
Hence, 
\[
\begin{split}
\sum_{x,y\in\mathbb{F}_q^*, x\neq y} |\chi_{\rho}(c_{x,y})|^2\chi_{U_{\alpha}}(c_{x,y})
&=\sum_{x,z\in\mathbb{F}_q^*, z\neq 1}\gamma(z)\alpha(x^2)\\
&=\left(\sum_{x\in\mathbb{F}_q^*}\alpha(x^2)\right)\left(\sum_{z\in\mathbb{F}_q^*, z\neq 1}\gamma(z)\right)=0\\
\end{split}
\]
by Corollary \ref{Cor:Character_0}, 
completing the proof for the case $\rho=W_{\beta,\beta'}$.

\paragraph{Case $\rho=V_{\beta}$.}
Since $|\chi_{V_{\beta}}(c_{x,y})|=1$ and $\chi_{U_{\alpha}}(c_{x,y})=\alpha(xy)=\alpha(x)\alpha(y)$, 
\[
\sum_{x,y\in\mathbb{F}_q^*, x\neq y} |\chi_{V_{\beta}}(c_{x,y})|^2\chi_{U_{\alpha}}(c_{x,y})
=\sum_{x,y\in\mathbb{F}_q^*, x\neq y}\alpha(x)\alpha(y) =\left(\sum_{x\in\mathbb{F}_q^*}\alpha(x)\right)^2 - \sum_{x\in\mathbb{F}_q^*}\alpha(x^2)=0
\]
by Lemma \ref{Lemma:Character_0} and Corollary \ref{Cor:Character_0}. This shows $C(V_{\beta},\alpha)=0$.

Now we are going to show that $D(V_{\beta},\alpha)=0$, or equivalently, $\sum_{x,y\in\mathbb{F}_q, y\neq 0} \alpha(\xi_{x,y}^{q+1})=0$. We have
\[
\sum_{\xi\in \mathbb{F}^*_{q^2}}\alpha(\xi^{q+1}) = \sum_{x,y\in\mathbb{F}_q, y\neq 0} \alpha(\xi_{x,y}^{q+1}) + \sum_{x\in\mathbb{F}^*_q} \alpha(\xi_{x,0}^{q+1})=\sum_{x,y\in\mathbb{F}_q, y\neq 0} \alpha(\xi_{x,y}^{q+1})\,.
\]
where in the last equality, we apply Corollary \ref{Cor:Character_0} and the fact that $\xi_{x,0}^{q+1}=x^{q+1}=x^2$ for all $x\in \mathbb{F}^*_q$.

Consider the map $\phi: \mathbb{F}^*_{q^2}\to\coms^*$ given by $\phi(\xi)=\alpha(\xi^{q+1})$. Clearly, $\phi$ is a character of $\mathbb{F}^*_{q^2}$. Since $\alpha^2$ is non-trivial and $\alpha^2(x)=\alpha(x^2)=\alpha(x^{q+1})=\phi(x)$ for all $x\in\mathbb{F}^*_q$, the map $\phi$ is also non-trivial. By Lemma \ref{Lemma:Character_0}, we have
\(
\sum_{\xi\in \mathbb{F}^*_{q^2}}\alpha(\xi^{q+1})=0\,,
\) 
which implies $D(V_{\beta},\alpha)=0$.

\paragraph{Case $\rho=X_{\varphi}$.}
As it is clear from the character table of $\GL_{2}(\mathbb{F}_q)$ that $C(X_{\varphi},\alpha)=0$, it remains to show $D(X_{\varphi},\alpha)=0$, or equivalently, $D_0\eqdef \sum_{x,y\in\mathbb{F}_q, y\neq 0}|\varphi(\xi_{x,y})+\varphi(\xi_{x,y}^q)|^2\alpha(\xi_{x,y}^{q+1})=0$. We have
\[
D_0 = \underbrace{\sum_{\xi\in \mathbb{F}^*_{q^2}}|\varphi(\xi)+\varphi(\xi^q)|^2\alpha(\xi^{q+1})}_{D_1} 
-\underbrace{\sum_{x\in\mathbb{F}^*_q}|\varphi(\xi_{x,0})+\varphi(\xi_{x,0}^q)|^2\alpha(\xi_{x,0}^{q+1})}_{D_2}\,.
\] 
For $\xi\in \mathbb{F}^*_{q^2}$, we have
\[
|\varphi(\xi)+\varphi(\xi^q)|^2 = (\varphi(\xi)+\varphi(\xi^q))(\varphi(\xi)^{-1}+\varphi(\xi^q)^{-1}) = 2+\varphi(\xi^{q-1})+\varphi(\xi^{1-q})\,.
\]
Hence, since $x^{q-1}=1$ for all $x\in\mathbb{F}_q^*$ and by Corollary \ref{Cor:Character_0},
\[
D_2= \sum_{x\in\mathbb{F}_q^*} (2+\varphi(x^{q-1})+\varphi(x^{1-q}))\alpha(x^{q+1}) = 3\sum_{x\in\mathbb{F}_q^*} \alpha(x^2)=0\,.
\]
The last thing we want to show is that $D_1=0$.
Consider the map $\phi: \mathbb{F}^*_{q^2}\to\coms^*$ given by $\phi(\xi)=\varphi(\xi^{q-1})\alpha(\xi^{q+1})$, which is apparently a character of $\mathbb{F}^*_{q^2}$. We shall see that it is non-trivial. Let $\omega$ be a generator of $\mathbb{F}^*_{q^2}$. Since $\omega^{q^2-1}=1$, we have $\phi(\omega^{q+1})=\alpha(\omega^{(q+1)^2})=\alpha(\omega^{2(q+1)})=\alpha^2(\omega^{q+1})$. On the other hand, $\omega^{q+1}$ is a generator for $\mathbb{F}^*_q$, because $\omega^{k(q+1)}$ with $k=0,1,\ldots, q-2$ are distinct, and $\omega^{(q-1)(q+1)}=1$. Hence, if $\phi(\omega^{q+1})=1$, then $\alpha^2(x)=1$ for all $x\in\mathbb{F}^*_q$. But since $\alpha^2$ is non-trivial, we must have $\phi(\omega^{q+1})\neq 1$, which means $\phi$ is non-trivial. Applying Lemma \ref{Lemma:Character_0}, we get
\(
\sum_{\xi\in \mathbb{F}^*_{q^2} }\varphi(\xi^{q-1})\alpha(\xi^{q+1}) =0\,.
\)
Similarly, we also have
\(
\sum_{\xi\in \mathbb{F}^*_{q^2} }\varphi(\xi^{1-q})\alpha(\xi^{q+1}) =0\,.
\)
Combining with the fact that \(
\sum_{\xi\in \mathbb{F}^*_{q^2}}\alpha(\xi^{q+1})=0\,,
\) which has been proved in the previous case, we have shown $D_1=0$, completing the proof.

\end{proof}


\end{document}